\documentclass{article}

\usepackage[utf8]{inputenc}
\usepackage{a4wide}

\usepackage[ruled,vlined]{algorithm2e}
\usepackage{lmodern}
\usepackage{color}
\usepackage[utf8]{inputenc}
\usepackage[T1]{fontenc}
\usepackage{latexsym,exscale,amssymb,amsmath, bm}
\usepackage{graphicx}
\usepackage[nointegrals]{wasysym}
\usepackage{eurosym}
\usepackage{hyperref}
\usepackage{amsthm}
\usepackage{caption}
\usepackage{blindtext}
\usepackage{bbm}
\usepackage{tikz}
\usepackage{subcaption}
\usepackage{enumerate}
\usepackage{enumitem}
\usepackage{natbib}
\usepackage{wrapfig}

\graphicspath{{images/}}

\definecolor{royalblue}{rgb}{0.2549019607843137, 0.4117647058823529, 0.8823529411764706}
\definecolor{crimson}{rgb}{0.8627450980392157, 0.0784313725490196, 0.23529411764705882}
\definecolor{forestgreen}{rgb}{0.13333333333333333, 0.5450980392156862, 0.13333333333333333}

\DeclareMathOperator*{\argmin}{argmin}

\DeclareMathOperator{\dir}{dir}
\DeclareMathOperator{\prj}{P}
\DeclareMathOperator{\prox}{prox}

\newcommand{\bhv}{\mathbb{T}_N}
\newcommand{\lnk}{\mathbb{L}_N}
\newcommand{\ort}{\mathbb{S}}
\newcommand{\perpdir}{\Sigma_{T}^\perp}
\newcommand{\prb}{P}
\newcommand{\prbalt}{Q}
\newcommand{\EE}{{\mathbb E}}
\newcommand{\NN}{{\mathbb N}}
\newcommand{\RR}{{\mathbb R}}
\newcommand{\TT}{{\mathbb T}}
\newcommand{\frf}{F_{\mathcal{T}}}
\newcommand{\mean}{\mu}
\newcommand{\smean}{\hat{\mean}}
\newcommand{\diff}{\mathrm{d}}
\newcommand{\stree}{\star}
\newcommand{\wst}[2][1]{\mathcal{P}^#1\left(#2\right)}
\newcommand{\iid}{{\stackrel{i.i.d.}{\sim}}}
\newcommand{\cO}{\mathcal{O}}

\newcommand{\cT}{\mathcal{T}}
\newcommand{\cE}{\mathcal{E}}
\newcommand{\cH}{\mathcal{H}}

\newtheorem{theorem}{Theorem}[section]
\newtheorem{proposition}[theorem]{Proposition}
\newtheorem{lemma}[theorem]{Lemma}
\newtheorem{corollary}[theorem]{Corollary}

\newtheorem{test}[theorem]{Test}

\theoremstyle{definition}
\newtheorem{definition}[theorem]{Definition}

\theoremstyle{remark}
\newtheorem{remark}[theorem]{Remark}
\newtheorem{example}[theorem]{Example}

\newcommand{\Comment}[1]{{$\star$\sf\textcolor{red}{#1}$\star$}}
\newcommand{\Commentl}[1]{{$\star$\sf\textcolor{magenta}{#1}$\star$}}
\SetKwComment{Commentc}{/* }{ */}

\definecolor{BrickRed}{cmyk}{0, .89, .94, .28}

\usetikzlibrary{shapes.geometric, arrows, arrows.meta}
\tikzstyle{process2} = [rectangle, 
minimum width=3.5cm, 
minimum height=1cm, 
text centered, 
text width=4cm, 
rounded corners=0.5cm,
draw=black, 
fill=lightgray!50]

\tikzstyle{arrow2} = [->, line width=.4mm]

\begin{document}
\title{Statistics for Phylogenetic Trees in the Presence of Stickiness}

\author{Lars Lammers, Tom M. W. Nye, Stephan F. Huckemann}
\maketitle

\tableofcontents

\begin{abstract}
Samples of phylogenetic trees arise in a variety of evolutionary and biomedical applications, and the Fr\'echet mean in Billera-Holmes-Vogtmann tree space is a summary tree shown to have advantages over other mean or consensus trees. 
However, use of the Fr\'echet mean raises computational and statistical issues which we explore in this paper. 
The Fr\'echet sample mean is known often to contain fewer internal edges than the trees in the sample, and in this circumstance calculating the mean by iterative schemes can be problematic due to slow convergence. 
We present new methods for identifying edges which must lie in the Fr\'echet sample mean and apply these to a data set of gene trees relating organisms from the apicomplexa which cause a variety of parasitic infections.
When a sample of trees contains a significant level of heterogeneity in the branching patterns, or topologies, displayed by the trees then the Fr\'echet mean is often a star tree, lacking any internal edges. 
Not only in this situation, the population Fr\'echet mean is affected by a non-Euclidean phenomenon called stickness which impacts upon asymptotics, and we examine two data sets for which the mean tree is a star tree. 
The first consists of trees representing the physical shape of artery structures in a sample of medical images of human brains in which the branching patterns are very diverse. 
The second consists of gene trees from a population of baboons in which there is evidence of substantial hybridization. 
We develop hypothesis tests which work in the presence of stickiness. 
The first is a test for the presence of a given edge in the Fr\'echet population mean; the second is a two-sample test for differences in two distributions which share the same sticky population mean. 
These tests are applied to the experimental data sets: we find no significant difference between male and female brain artery tree populations; in contrast, significant differences are found between subgroups of slower- and faster-evolving genes in the baboon data set.  
\end{abstract}

\section{Introduction}

The Billera-Holmes-Vogtmann (BHV) phylogenetic tree spaces are a class of metric spaces of phylogenetic trees, initially proposed in \cite{bhv}.
Phylogenetic trees are edge-weighted trees whose leaves represent present-day taxa and whose branching structure reflects the shared ancestry of taxa.
They can also be used to represent the physical shape of branching structures such as blood vessels. 
If every branch point in a tree is binary, the tree is called \emph{resolved}, and otherwise the tree is \emph{unresolved}.
The BHV tree space $\bhv$ is the set of all resolved and unresolved edge-weighted trees with leaves bijectively labelled $1,\ldots,N$. 
It is a stratified space, with one stratum for each tree topology; the topology of a tree is its structure modulo edge weights. 
As metric spaces, the BHV tree spaces have attractive geometrical properties: for example, the existence and uniqueness of a geodesic between any given pair trees, and convexity of the metric along geodesics.
Subsequently, an algorithm for computing geodesics in polynomial time \citep{OwenProvan11} with respect to $N$ was published.
Put together, the geometric properties of $\bhv$ and efficient computation of geodesics has enabled the development of a range of statistical methods for analysing samples of trees in BHV tree space.
Alternative spaces for phylogenetic trees have been proposed such as \emph{tropical tree space} by \cite{LinMonodYoshida2018} exploiting computational feasibility of tropical geometry \citep{MaclaganSturmfels2015} and \emph{wald space} by \cite{garba2021information, lueg2024foundations} which incorporates features of the models used to infer trees from genetic data.
However, BHV tree space has attracted the most statistical development to date due to its unique features. 

Specific methods developed for statistics in $\bhv$ include computational methods for
Fr\'echet means in BHV spaces \citep{bhvmean}, and principal component analysis \citep{bhvpca1, feragen13, bhvpca2}.
The Fr\'echet mean is a generalization of the mean of a probability distribution to metric spaces, defined in the following way. 
Given a probability distribution $\prb$ on a metric space $(M,d)$, the Fr\'echet function is defined as
\begin{align}\label{eq:Frechetfcn}
	F_\prb(x) = \frac{1}{2} \int_M d^2(x,y) \ \diff \prb(y)\,, \quad x \in M\,,
\end{align}
if the integral exists. 
The Fr\'echet mean of $\prb$ is then given by
\begin{align}\label{eq:Frechetmean}
	\mathfrak{b}(\prb) = \argmin_{x \in M} F_\prb(x)\,.
\end{align}
As \citet[p. 33]{sturm} noted, for the existence of Fr\'echet means on a complete metric space $(M,d)$ it suffices to require that $\prb\in \wst[1]{M}$ where
\begin{align}\label{eq:prob-spaces}
    \wst[k]{M} &:= \left\{ \prb \in \mathcal{P}(M) : \exists x \in M : \int_{M} d(x,y)^k \ \diff \prb(y)\right\}\,, k\in \NN
\end{align}
and $\mathcal{P}(M)$ denotes the set of Borel-probability distributions on $M$.
As established \cite[Lemma 4.1]{bhv}, BHV spaces are Hadamard spaces, ensuring the uniqueness of Fr\'echet means \cite[Proposition 4.3]{sturm}. 
Moreover, in \cite{brown2018mean} simulations were used to show that the Fr\'echet sample mean in BHV space offers advantages over mean trees or consensus trees defined in a different way.

However, use of the Fr\'echet mean is not without issues. 
On the one hand, if the topology of the mean is known, then there are algorithms which quickly determine the edge weights in the mean tree \cite{skwerer2018relative}. 
On the other hand, when the topology is unknown, iterative algorithms can be used to find the mean tree topology and edge weights \citep{bacak_alg,miller2015polyhedral,sturm}, but these algorithms can converge quite slowly. 
In particular, it has been observed that when the mean tree is less resolved than the trees in the sample, then the iterative algorithms keep changing the topology of the estimate of the mean, even after many iterations. 
In fact, we show in Theorem~\ref{thm:sturms-algo-not-singular} that under certain common conditions, Sturm's algorithm \citep{sturm} almost surely changes topology an infinite number of times. 
The main computational issue and open problem is therefore how to determine the topology of the Fr\'echet mean. 

In practice, a pragmatic approach has been to use an iterative algorithm for computing the Fr\'echet mean such as Sturm's algorithm, but then to ignore very short edges that come and go as the algorithm proceeds and as the topology repeatedly changes.  
In a similar way to a criterion due to~\cite{bhvmean}, in Theorem~\ref{thm:main-strata} we provide a sufficient condition for the presence of an edge in the Fr\'echet mean tree, and a related algorithm (Algorithm~\ref{alg:find-splits}) for finding these edges. 
Due to decomposition in Theorem \ref{thm:tancone}, this criterion involves directional derivatives of the Fr\'echet function~$\eqref{eq:Frechetfcn}$ at the star tree. 
Since the algorithm is not guaranteed to find all the edges in the Fr\'echet mean, information from directional derivatives can be used to add further edges to the topology of the proposed mean tree, which is the idea behind our Algorithms~\ref{alg:prox} and~\ref{alg:mindegrees}. 
Specifically, we minimize directional derivatives orthogonal to the topology of the proposed mean, and iteratively add edges to this topology. 

The second challenge presented by the BHV Fr\'echet mean is of statistical nature.
The asymptotic behavior of the sample Fr\'echet mean deviates from the classical case in Euclidean spaces. 
For some distributions, the sample Fr\'echet mean will be almost surely confined for large sample sizes to certain lower-dimensional strata of the BHV space containing unresolved trees \cite{barden18}.
This phenomenon is referred to as \emph{stickiness} of the Fr\'echet mean \cite{openbook, kale}
and poses new challenges.
For example, given two distributions with Fr\'echet means on the same lower dimensional stratum, the effectiveness of the sample Fr\'echet mean to discriminate between the two distributions is reduced. 
This is particularly problematic in the context of hypothesis testing.

Since our criterion for presence of an edge in the Fr\'echet mean tree applies to distributions as well as finite samples, we are able to define hypothesis tests for the population mean tree in the presence of stickiness. 
Given a proposed topology for the mean tree, we describe a hypothesis test for the presence of a single edge in the population mean tree and a joint test for multiple edges. 
Finally we propose a two-sample hypothesis test for equality of two distributions when they have the same sticky Fr\'echet population mean. 

We apply these algorithms and tests to experimental data from three studies with important biomedical applications: evolutionary trees for a set of apicomplexa (single-celled organisms associated with several parasitic infections) \citep{api1}; trees representing physical blood vessel structure in human brains \citep{skwerer2014tree}; and thirdly a large data set of evolutionary trees for a number of related populations of baboon, each tree corresponding to a different genomic locus \citep{baboons}. 
For the data set of apicomplexan gene trees, we confirm a previously published unresolved topology for the Fr\'echet mean which relied on arbitrary removal of short edges. 
For the latter two data sets the Fr\'echet sample mean is entirely unresolved, thereby representing a particular challenge to statistical analysis. 
We compare male and female artery trees: tests confirm that the Fr\'echet mean is the star tree for both populations, and moreover, that no significant difference can be detected between the two distributions. 
The Fr\'echet sample mean for the baboon trees is the star tree. 
In \cite{baboons} it was proposed that a high level of hybridization was present in the baboon populations. 
This non-vertical ancestry of genetic material causes a high level of topological heterogeneity in the trees from different loci, and as a result the fully-unresolved mean tree. 
Nonetheless, significant differences are found between subgroups of slower- and faster-evolving genes.

\begin{figure}
    \label{fig:enter-label}
    \centering
	\begin{tikzpicture}[node distance=2cm]
		\node (sample) [process2, ] {sample $\cT\subset \bhv$};
		\node (proposal) [process2,below of=sample, ] {splits in sample Fr\'echet mean};
		\node (verify) [process2, below of=proposal,] {verify proposed topology};
		\node (test) [process2, below of=verify,xshift = -2.5cm] {presence of splits in population Fr\'echet mean};
        \node (twosample) [process2, right of=test, xshift = 2.5cm] {discriminate between sticky distributions};

		\draw [arrow2, black] (sample) to node[anchor= west, xshift=.0cm,align=center] {Algorithm~\ref{alg:find-splits}}(proposal);
		\draw [arrow2, black] (proposal) to node[anchor= west, xshift=.0cm,align=center] {Algorithms\\\ref{alg:mindegrees} and \ref{alg:prox}}(verify);
		\draw [arrow2, black] (verify) to node[anchor= east, xshift=.0cm,align=center] {Test~\ref{test:one-sample} or \ref{test:one-sample-strata}}(test);
		\draw [arrow2, black] (verify) to node[anchor= west, xshift=.0cm,align=center] {Test~\ref{test:two-sample-star} or \ref{test:two-sample-strata}}(twosample);
	\end{tikzpicture}
	\caption{\it Our proposed tool chain for finding splits in the topology of the sample Fr\'echet mean and our new hypothesis tests for population  Fr\'echet means.}
	\label{fig:proc-chart}
\end{figure}
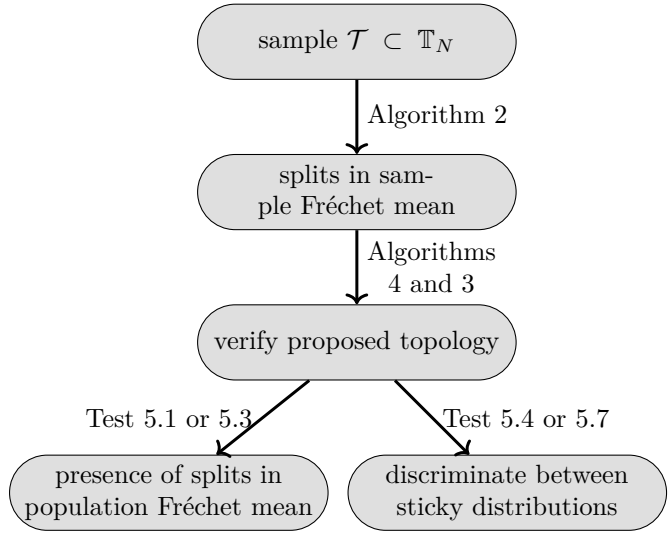

Figure \ref{fig:proc-chart} depicts the new tool chain based on our paper, which is structured as follows. In Section \ref{sec:BHV}, we establish
notation for BHV spaces and give a brief summary of results on geometry
of BHV spaces and the Fr\'echet mean.
In Sections \ref{sec:mean_top} and \ref{sec:min_dir_deriv}, we present our 
contributions to finding
the topology of the mean tree. Section \ref{sec:testing} is concerned
with possible applications in hypothesis testing.
Applications to experimental data are presented in Section \ref{sec:examples}.
All proofs are given in the Appendix \ref{sec:proofs}.

\section{Billera-Holmes-Vogtmann Phylogenetic Tree Spaces}
\label{sec:BHV}

\subsection{Phylogenetic Trees and Notation}

We call a directed acyclic graph with non-negatively weighted edges a \emph{phylogenetic tree} if every internal node is of degree of at least 3 and each exterior node is assigned a unique label, one of which is designated as \emph{root}. The remaining exterior nodes are referred to as \emph{leaves}. The weights of the edges are regarded as their lengths.

Edges of such a tree can be characterized by \emph{splits} and we characterize only \emph{interior} edges so. The removal of an internal edge results in the split of labels into disjoint sets $A$ and $B$, each of them containing at least two elements. We then write $s = A \vert B = B \vert A$ for that split. Two splits $A \vert B$ and $C\vert D$ are said to be \emph{compatible }if at least one of the intersections $A\cap C$, $A\cap D$, $B \cap C$ or $B \cap D$ is empty. Otherwise, we call them incompatible. Note that two incompatible splits cannot be present in the same tree.

The \emph{topology} of a phylogenetic tree is then given by the set of its splits. A topology of a tree with $N$ leaves is binary if it contains $(N-2)$ splits (interior edges). A tree which is binary is also called \emph{fully resolved}. A \emph{star tree} has no splits. 
Two examples of phylogenetic trees are displayed in Figure
\ref{fig:phylotrees}.

\begin{figure}[!h]
	\centering
	\subfloat[{\it A star tree, i.e. a tree without splits (interior edges).}]{
		\includegraphics*[width=5cm, height=5cm]{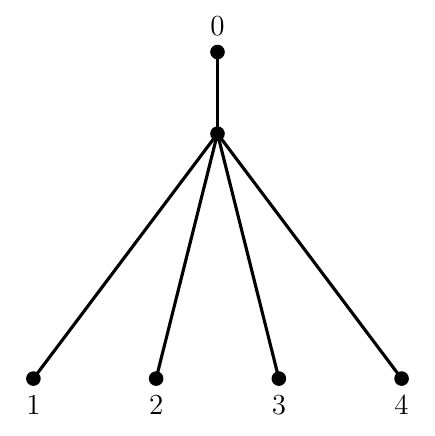}}
	\qquad
	\subfloat[{\it A binary tree with two splits \textcolor{royalblue}{$0,4 \vert 1,2,3$} and \textcolor{crimson}{$0,3,4 \vert 1,2$}.}]{
		\includegraphics*[width=5cm, height=5cm]{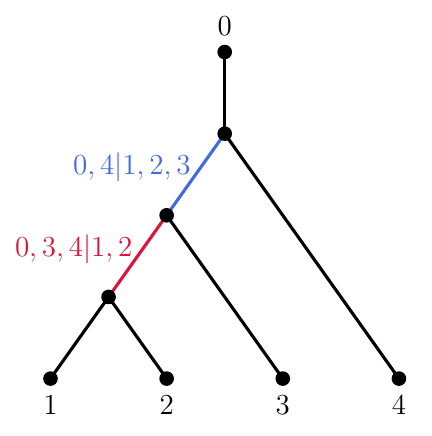}}
	\caption{\it Two phylogenetic trees with four leaves.}
	\label{fig:phylotrees}
\end{figure}

Throughout this paper, we shall use the following notation.
\begin{definition}\label{def:foundations}
	Suppose $N\geq 3$ leaves are given. Let $s$ be a split and $T$ be a
	phylogenetic tree.
	\begin{enumerate}
		\item $E(T)$ denotes the set of all splits of $T$. 
		\item $C(s)$ denotes the set of splits compatible with $s$ 
  and 
        $C(T) = \bigcap_{x \in E(T)} C(x)$ the set of splits compatible with
		      the topology of $T$ 
		\item Two sets $A, B$ of splits are compatible if all splits are pairwise
		      compatible.
		\item If $s \in C(T)\setminus E(T)$, write $T + \lambda \cdot s$ for the tree obtained by adding the split
		      s to $T \in \bhv$ with length $\lambda >0$ and set $T + 0 \cdot s = T$.
		\item $\lvert s \rvert_T$ denotes the length of $s$ in $T$ (possibly 0 if the split is not present) and write $$\rVert T \lVert := \sqrt{\sum_{x \in E(T)} \lvert x \rvert_T^2}\,.$$
	\end{enumerate}
\end{definition}

\subsection{Construction of BHV Spaces: Stratification and the Star Tree $\star$}\label{scn:constructing-bhv}

Next, we briefly introduce Billera-Holmes-Vogtmann tree spaces, giving only essential details. For a formal descirption,
see \cite{bhv}.
BHV spaces take split (interior edges) lengths into account, they ignore pendant edges lengths. For $N \in \mathbb{N}$ and $N \geq 3$ there are $M:=2^{N-1}-N-1$ possible splits, so the BHV space $\bhv$ of phylogenetic trees with $N$ leaves is the following subset of $\RR^M$: each topology is assigned a corresponding \emph{orthant} and within, the positive coordinates of a tree with that topology are given by the lengths of its splits. Binary trees are assigned corresponding $(N-2)$-dimensional
top-dimensional orthants and non-binary topologies occur at the boundaries of multiple top-dimensional orthants.
these orthants. In particular, non-binary topologies occur at the boundaries of multiple orthants. Thus the metric space $(\bhv,d_{\mathbb T_N})$ arises from embedding in the Euclidean $\RR^M$, i.e.  through gluing the orthants together at these boundaries and carrying the induced intrinsic metric.

\begin{definition}
	\label{def:stratum}

    For $N \geq 4$ and $1\leq l \leq N-3$, trees with common topology  given by the same $N-2-l$ splits, form a  \emph{stratum $\ort \subset \bhv$ of codimension $l$}. The metric projection onto its closure $\overline{\ort}$ is denoted by 
    $$\prj_{\overline{\ort}}: \bhv \to \overline{\ort}, \quad T \mapsto \argmin_{T^\prime \in \overline{\ort}} d_{\bhv}(T, T^\prime)\,.$$ 
 
\end{definition}

\begin{remark}
   Every stratum of codimension $1\leq l \leq N-2$, as it corresponds to a $(N-l-2)$-dimensional Euclidean orthant, is convex. In consequence the projection $\prj_{\overline{\ort}}$ to a closed convex set is well defined.
\end{remark}

This construction of \cite{bhv} 
yields a complete metric space $(\bhv,d)$ of global non-positive curvature -- a \emph{Hadamard space}.

All orthants are open, only the star tree is a single point in BHV space forming the stratum of codimension $N-2$, namely the origin of $\RR^M$. It also acts as the 'origin' of BHV space, as we will see in Section \ref{sec:directions}. In this work, the
star tree plays a key role, and will be denoted by $$\stree \in \bhv\,,$$
with trivial metric projection
 $$\prj_{\{\stree\}}: \bhv \to \{\stree\}, \quad T \mapsto \stree\,.$$

\subsection{Geodesics and Geodesics in BHV spaces}

\begin{definition}

	Let $(M,d)$ be a metric space. A  continuous curve $\gamma: [a,b] \to M$ is a \emph{geodesic} if there is a constant $c >0 $, called its \emph{speed}, such that for every $\kappa \in (a,b)$, there is $\epsilon > 0$ such that for every $\lambda_1, \lambda_2 \in [\kappa-\epsilon, \kappa+\epsilon]$
	\begin{align*}
		d(\gamma(\lambda_1), \gamma(\lambda_2)) = c \lvert \lambda_1 - \lambda_2 \rvert
	\end{align*}
 
	A geodesic is called \emph{minimizing} if the equation above holds for every $\lambda_1, \lambda_2$.
	We say a geodesic is of \emph{unit speed} if $c=1$.
 
	A \emph{geodesic segment} $G \subseteq M$ is the image of a geodesic, meaning there 	is a geodesic $\gamma: [a,b] \to M$ such that $G = \gamma([a,b])$.

\end{definition}

Let $\gamma: [a,b] \to \bhv$ be a continuous curve. Then its length  in the above introduced metric  $d_{\bhv}$ is given by
\begin{align*}
	L\left(\gamma\big\vert_{[\lambda_1,\lambda_2]}\right) = \sup_{\substack{ \lambda_1 \leq t_1 < \ldots < t_n \leq \lambda_2 \\
			n \in \mathbb{N}\setminus \{1\}
		}} \sum_{i=1}^{n-1} \lVert\gamma(t_{i}) - \gamma(t_{i+1})\rVert_2,
\end{align*}
where $\gamma_(t_i)$ and $\gamma_(t_i+1)$ are chosen to lie in the same closed orthant in which $ \lVert \cdot \lVert_2$ denotes the Euclidean norm.
Then, for $T_1, T_2 \in \bhv$, 
their distance in $\bhv$ is also given by
\begin{align*}
	d_{\bhv}(T_1, T_2) = \min_{\footnotesize \begin{array}{c} 
 \gamma: [a,b] \to \bhv\mbox{ continuous }\\\gamma(a) =T_1,  \gamma(b) = T_2
 \end{array}} L(\gamma)\,.
\end{align*}


Since BHV spaces are Hadamard spaces, as mentioned at the end of the previous Section \ref{scn:constructing-bhv}, geodesics are unique and so are Fr\'echet means \citep{sturm}.

With the notion of `support pairs'  \cite{owenphd} 
characterize geodesics and  \cite{OwenProvan11} develop methods to compute geodesics in BHV spaces in
polynomial time.

\begin{definition}

	Let $N \geq 3$ and $T_1, T_2 \in \bhv$. Let $\mathcal{A} = (A_0, A_1,
		\ldots A_k)$ be a partition of $E(T_1) \cup (C(T_1) \cap E(T_2))$ and
	$\mathcal{B} = (B_0, B_1, \ldots B_k)$ be a partition of $E(T_2) \cup
		(C(T_2) \cap E(T_1))$.

	The pair $(\mathcal{A}, \mathcal{B})$ is called a \emph{support pair} of
	$T_1$ and $T_2$ if
	\begin{enumerate}

		\item $A_0$ and $B_0$ are equal and contain all splits
		      that are either shared between $T_1$ and $T_2$ or compatible
		      with the other tree, i.e.
		      \begin{align*}
			      A_0 = B_0 = \left( E(T_1) \cap E(T_2) \right) \cup \left( E(T_1) \cap C(T_2) \right) \cup \left( E(T_2) \cap C(T_1)\right)\,,
		      \end{align*}
		      and
		\item $A_i$ is compatible with $B_j$ for all $1 \leq j < i \leq k$.
	\end{enumerate}
\end{definition}

\cite{Owen2011only} showed that the geodesic segment joining any two points $T_1, T_2$ corresponds to a unique support pair of the two trees $T_1, T_2$ with the additional property
\begin{eqnarray}\label{eq:bhv-geodesics1}\frac{\|A_i\|}{\|B_i\|}&\leq&  \frac{\|A_{i+1}\|}{\|B_{i+1}\|},\quad 1\leq i < k\,.\end{eqnarray}
Here,
$$ \|A_i\|:= \|A_i\|_{T_1} = \sqrt{\sum_{s\in A_i}\vert s\vert^2_{T_1}}\quad\mbox{ and }\quad\|B_i\| := \|B_i\|_{T_2} =\sqrt{\sum_{s\in B_i}\vert s\vert^2_{T_2}}\quad\mbox{ for }1\leq i \leq k\,.$$

Then the geodesic $\overline{\gamma}: [0,1] \to \bhv$ with $\overline{\gamma}(0) = T_1$ and
$\overline{\gamma}(1)=T_2$ can be split into the following segments

\begin{align*}
	\overline{\gamma}_i = \begin{cases}
		           \overline{\gamma}\left(\left[0, \frac{\lVert A_1 \rVert}{\lVert A_1 \rVert + \lVert B_1 \rVert}\right)\right)                                                                           & \quad \text{if }i =0\,,         \\
		           \overline{\gamma}\left(\left[\frac{\lVert A_i \rVert}{\lVert A_i \rVert + \lVert B_i \rVert}, \frac{\lVert A_{i+1} \rVert}{\lVert A_{i+1} \rVert + \lVert B_{i+1} \rVert}\right)\right) & \quad \text{if }1 \leq i < k\,, \\
		           \overline{\gamma}\left(\left[\frac{\lVert A_i \rVert}{\lVert A_i \rVert + \lVert B_i \rVert}, 1\right]\right)                                                                           & \quad \text{if }i =k\,,         \\
	           \end{cases}
\end{align*}
such that $\overline{\gamma}_i$ is a straight line in the closed orthant corresponding to the splits of
$A_0\cup B_1 \cup \ldots B_i \cup A_{i+1} \cup \ldots A_k$.

The length of a split $s$ at $\lambda \in [0,1]$ is then given by
\begin{align}
\label{eq:bhv_geod_splits}
	\lvert s \rvert_{\overline{\gamma}(\lambda)} = \begin{cases}
		                                    (1 - \lambda) \lvert s \rvert_{T_1} + \lambda \lvert s \rvert_{T_2} \quad                                          & \text{if } s \in A_0=B_0\,, \\
		                                    \frac{(1 - \lambda) \lVert A_i \rVert - \lambda \lVert B_i \rVert}{\lVert A_i \rVert} \lvert s \rvert_{T_1}
         \cdot 1_{\left[0,\frac{\|A_i\|}{\|A_i\|+\|B_i\|}\right]}(\lambda)         
                                      \quad  & \text{if } s \in A_i\,,1\leq i\leq k\,,     \\
		                                    \frac{ \lambda \lVert B_i \rVert - (1 - \lambda) \lVert A_i \rVert}{\lVert B_i \rVert} \lvert s \rvert_{T_2}
                                               \cdot 1_{\left[\frac{\|A_i\|}{\|A_i\|+\|B_i\|},1\right]}(\lambda) 
\quad & \text{if } s \in B_i\,,1\leq i\leq k\,.     \\
	                                    \end{cases}
\end{align}
In total, one obtains for the geodesic distance of $T_1$ and $T_2$
\begin{align}\label{eq:bhv-distance-explicit}
	d_{\bhv}(T_1, T_2) = \sqrt{\sum_{s \in A_0}  (\lvert s \rvert_{T_1} - \lvert s \rvert_{T_1})^2 + \sum_{i =1}^k(\lVert A_i \rVert + \lVert B_i \rVert)^2}\,.
\end{align}

Vice versa, \cite{OwenProvan11} showed that the above construction leads to a geodesic if and only if
\begin{eqnarray}\label{eq:bhv-geodesics2}
 \frac{\|I_1\|}{\|J_1\|}&>&  \frac{\|I_{2}\|}{\|J_{2}\|}
\end{eqnarray}
for all $i \in \{1,2 \ldots, k\}$ and every nontrivial (i.e. each containing at least one element)
		      partitions $A_i = I_1 \cup I_2$, $B_i = J_1 \cup J_2$ such that
		      $I_2$ and $J_1$ are compatible.





An illustration of geodesics is given in Figure \ref{fig:bhv_geods}.

\begin{figure}[!ht]
    \centering
    \includegraphics[scale=.6]{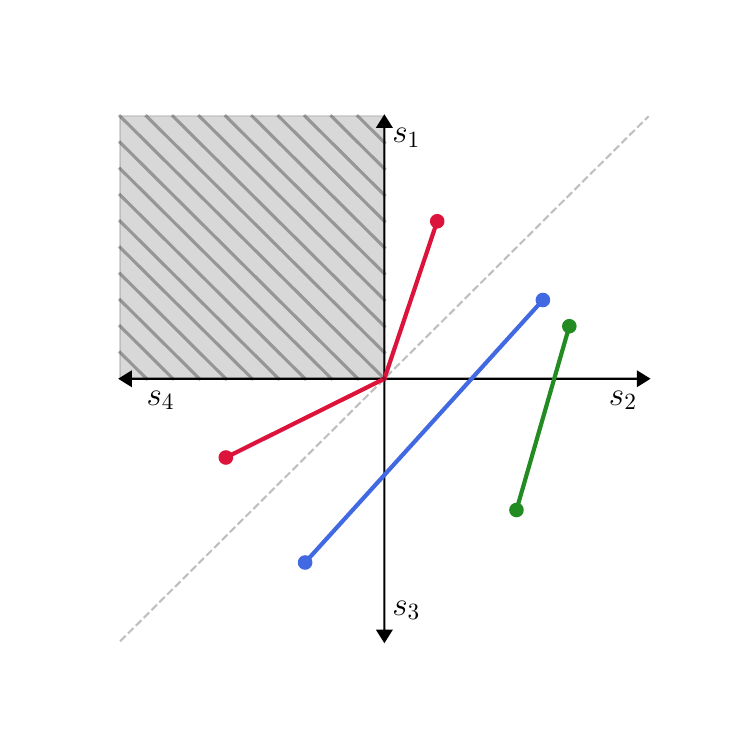}
    \caption{\it Three examples of geodesics in a subset of $\mathbb{T}_4$ from the top point to the bottom point, respectively. 
    The corresponding support pairs are \textcolor{crimson}{$\big((\emptyset, \{s_1,s_2\}),(\emptyset, \{s_3,s_4\})\big)$},
    \textcolor{royalblue}{$\big((\emptyset, \{s_1\},\{s_2\}),(\emptyset, \{s_3\},\{s_4\})\big)$} and \textcolor{forestgreen}{$\big((\{s_2\}, \{s_1\}),(\{s_2\}, \{s_3\})\big)$}.
    }
    \label{fig:bhv_geods}
\end{figure}

\begin{definition}\label{def:bhv-geodesics}

	Let $N \geq 3$. For $T_1, T_2 \in \bhv$, we denote by $\gamma_{T_1}^{T_2}:
			      [0, d(T_1, T_2)] \to \bhv$ the unique unit speed geodesic from $T_1$ to $T^\prime$.
		      The geodesic reparametrized on the unit interval is denoted by $\overline{\gamma}_{T_1}^{T_2} : [0,1] \to \bhv\,, t \mapsto \gamma_{T_1}^{T_2}(t \cdot d(T_1, T_2))$.
\end{definition}

\begin{remark}
    \label{rem:dist_single_split}

    Let $T \in \bhv$, $N \geq 3$, and $s$ be an arbitrary split of the appropriate labels.
    Then, the computation of the distance between $T$ and the tree $\stree + \lambda \cdot s$, $\lambda >0$, 
    is rather simple, as is easily verified by \eqref{eq:bhv-distance-explicit}.
    We have
    \[
        d^2(T,\stree + \lambda \cdot s) = \begin{cases}
             \left(\lambda + \sqrt{\sum_{C(s)\not\ni x \in E(T)}
				\lvert x \rvert_T^2}\right)^2 + \sum_{x \in E(T) \cap C(s)} \lvert x \rvert_T^2        &\text{if } s \not\in C(T)\,,  \\
             (\lambda - \lvert s \rvert_T)^2 + \sum_{s\neq x \in E(T)} \lvert x \rvert_T^2 \quad &\text{if } s \in E(T)\,,\\
             \lambda^2 + \sum_{x \in E(T)} \lvert x \rvert_T^2 &\text{else. }\\ 
             \end{cases}
    \]
\end{remark}

\subsection{The Space of Directions and the Tangent Cone}
\label{sec:directions}
Our results and proposed methods rely on the notion of directions of geodesics. For the following 
and further reading, we refer to \cite{barden18} for orthant spaces, a slightly more general notion of BHV spaces, and \cite[Chapter 3, 9]{burago} for a broader overview in metric geometry.

\begin{definition}

	Let $N \geq 3$ and let $T \in \bhv$ and let $\gamma: [0,a] \to \bhv$,
	$\gamma^\prime:[0,a^\prime] \to \bhv$ be two 
	$\gamma(0) = \gamma^\prime(0) = T$.
	\begin{enumerate}
		\item The \emph{Alexandrov angle} at $x$ between $\gamma$ and $\gamma^\prime$ is
		      given by
		      \begin{align*}
			      \angle_T(\gamma, \gamma^\prime) := \lim_{\lambda, \lambda^\prime \searrow 0} \arccos
			      \left( \frac{\lambda^2+ {\lambda'}^2 -d(\gamma(\lambda), \gamma^\prime(\lambda^\prime))}
			      {2 \lambda \cdot \lambda^\prime}\right)\,.
		      \end{align*}
		\item If $\angle_T(\gamma, \gamma^\prime) = 0$, the two geodesics have the \emph{same direction} $\sigma$ at $T$. The \emph{space of directions at $T$}, denoted by $\Sigma_T$, is given by the set of equivalence classes of geodesics with the same directions.
		\item The space $\mathfrak{T}_T = \Sigma_T \times \RR_0 / \sim$, where $(\sigma_1, r_1) \sim  (\sigma_2, r_2)$ if
		      $(\sigma_1 = \sigma_2 \land r_1 = r_2)$ or $(r_1=0=r_2)$, is called 		      the \emph{tangent cone} at $T$ and it is equipped with the distance
		      \begin{align*}
			      \tilde{d}_T((\sigma_1, r_1), (\sigma_2, r_2)) := \sqrt{r_1^2 + r_2^2 - 2 r_1 r_2 \cos(\angle_x(\sigma_1, \sigma_2))}\,.
		      \end{align*}
		      The class of $(\sigma, 0)$ is called the \emph{cone point} $\cO$.

	\end{enumerate}

\end{definition}

\begin{remark}
    \label{rem:al_angle}

    By construction, the Alexandrov angle confined to the range $[0,\pi]$.
\end{remark}

The construction of the spaces of directions and the tangent cone result in two metric spaces.

\begin{proposition}

	Let $N \geq 3$ and $T \in \bhv$. Then, both $(\Sigma_T, \angle_T)$ and
	$(\mathfrak{T}_T, \tilde{d}_T)$ are complete metric spaces.

\end{proposition}

From \cite[Theorem II.3.19]{bridson} we take the following desirable properties.

\begin{remark}\label{rm:Cat-1}
    For any $T\in \bhv$, in particular, $(\Sigma_T, \angle_T)$ is a CAT(1) space, implying that there are unique geodesics between points of distance $<\pi$. Moreover, $(\mathfrak{T}_T, \tilde{d}_T)$ is a Hadamard space.
\end{remark}


\begin{definition}\label{def:geodesics}

	Let $N \geq 3$.
	\begin{enumerate}
		\item Let $T \in \bhv$. For two directions $\sigma, \sigma^\prime \in \Sigma_T$ with $\angle_T(\sigma, \sigma^\prime) < \pi$,
		      we denote the unit speed geodesic from $\sigma$ to $\sigma^\prime$ by $\beta_\sigma^{\sigma^\prime} : [0, \angle_T(\sigma, \sigma^\prime)] \to \Sigma_T$.
		      We write $\overline{\beta}_\sigma^{\sigma^\prime} : [0,1] \to \Sigma_T\,, t \mapsto \beta_\sigma^{\sigma^\prime}(t \cdot \angle_T(\sigma, \sigma^\prime))$ for the reparametrized geodesic over the unit interval.

		\item For $T \in \bhv$, let $\dir_T: \bhv\setminus\{T\} \to \Sigma_T$ denote the map
		      that maps $T^\prime \in \bhv\setminus\{T\}$ to the direction of the geodesic $\gamma_T^{T^\prime}$ at $T$.
		\item The \emph{logarithm map} at $T \in \bhv$ is defined as
		      \begin{align*}
			      \log_T :  \bhv & \to \mathfrak{T}_T                                                             \\
			      T^\prime       & \mapsto \begin{cases}
				                               \mathcal{O}                        & \quad \text{if $T = T^\prime$,} \\
				                               (\dir_T(T^\prime), d(T, T^\prime)) & \quad \text{else.}
			                               \end{cases}
		      \end{align*}
	\end{enumerate}
\end{definition}

In view of the above definition of the logarithm  we can simply use points in $\bhv$ as directions, representing unit speed geodesics from another point in  $\bhv$ . For instance for $T_1, T_2 \in \bhv$, $T_1\neq T_2$ and $\sigma \in \Sigma_{T_1}$, we write
\[
    \angle_{T_1} (\sigma, T_2) := \angle_{T_1} (\sigma, \dir_{T_1}(T_2))\,.
\]
Furthermore, if T is not fully resolved, we write $\sigma_s = \dir_T(T + 1\cdot s) \in \Sigma_T$
for a split $s \in C(T)\setminus E(T)$.

\begin{figure}[!h]
	\centering
	\subfloat[{\it The Peterson graph.}]{
		\includegraphics*[width=7cm, height=7cm]{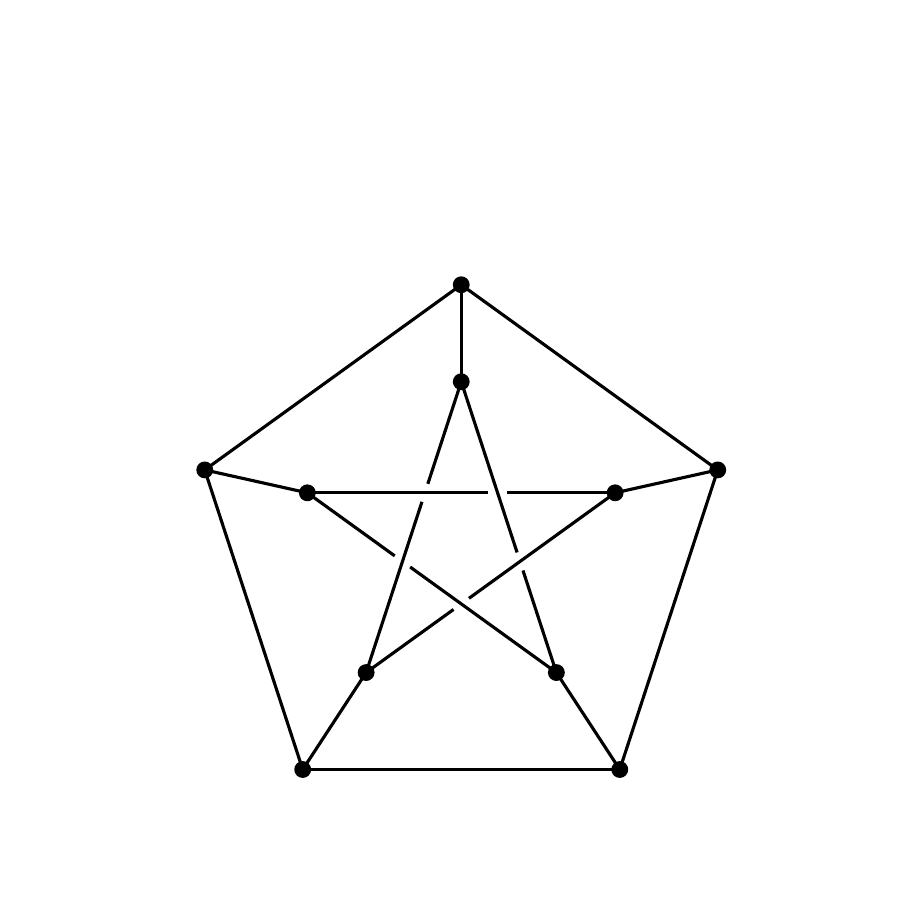}}
	\subfloat[{\it Three maximal orthants of $\mathbb{T}_4$ (gray) with corresponding part of the  link (black).}]{
		\includegraphics*[width=6cm, height=6cm]{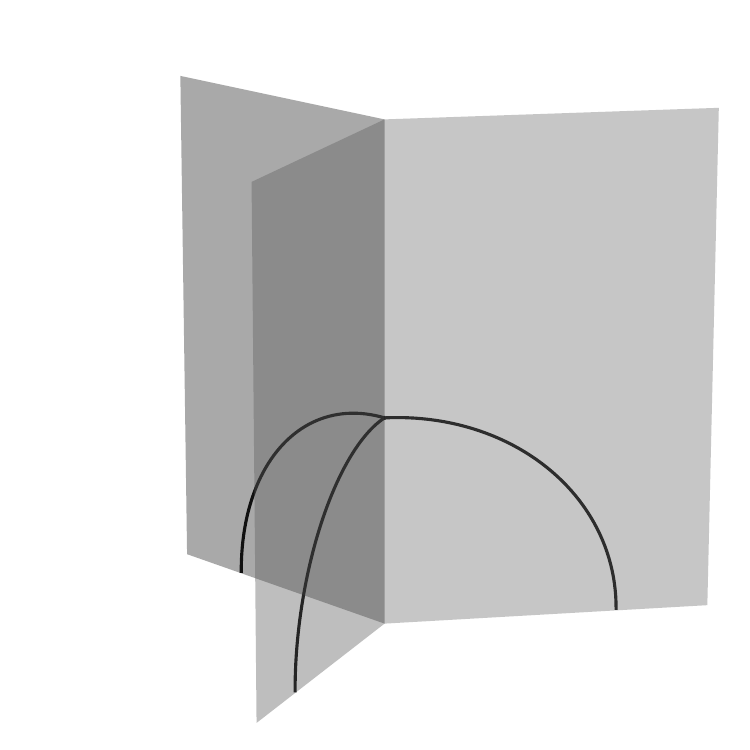}}
	\caption{\it The link $\mathbb{L}_4$ is the Peterson graph, in which each
		edge has length $\pi/2$.}
	\label{fig:BHVlink}
\end{figure}

It turns out that the space of directions at the star tree can be identified with the 
\emph{link}, 
\begin{align*}
	\lnk := \{ T \in \bhv \ : \ \lVert T \rVert = 1 \}\,,
\end{align*}
the unit sphere in BHV spaces.
A part of the link in $\mathbb{T}_4$ is depicted in Figure \ref{fig:BHVlink}.

In \cite[Section 4.1]{bhv}, an alternative way of constructing BHV spaces as cones over the respective links is discussed, with the star tree serving as the cone point. In consequence, the tangent cone at $\stree$ is just the space itself, and the link can be identified with the space of directions at $\stree$. We condense this in the following proposition.

\begin{proposition}
    \label{prop:bhv_cone}

	Let $N \geq 3$. Then, the following hold.
	\begin{enumerate}
		\item The map $\mathbb{L}_N \to \Sigma_\stree, T \mapsto \dir_\stree(T)$
		      is \emph{bijective}.
		\item The map $\log_\stree: \bhv \to \mathfrak{T}_\stree$ is an \emph{isometry}.
	\end{enumerate}
	Furthermore, one has for any $T_1, T_2 \in
		\bhv\setminus \{\stree\}$ that
	\begin{align*}
		\angle_\stree\left(\dir_\stree(T_1),\dir_\stree(T_2)\right) = \arccos\left(\frac{\lVert T_1 \rVert^2 + \lVert T_2 \rVert^2 -  d^2(T_1, T_2)}{2 \lVert T_1 \rVert \lVert T_2 \rVert}\right)\,.
	\end{align*}

\end{proposition}
\begin{figure}
	\centering
	\includegraphics[scale=.75]{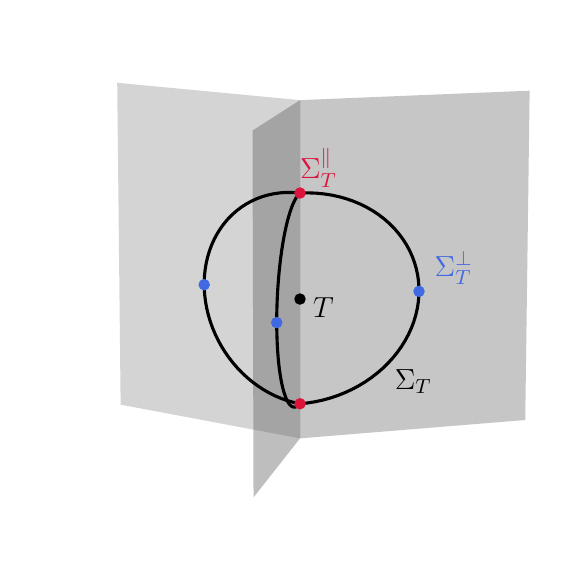}
	\caption{\it The space of directions at a point $T$ lying in a stratum of codimension 1
		in $\mathbb{T}_4$. There are two directions parallel and three directions
		perpendicular to the stratum in which $T$ resides.}
	\label{fig:directions}
\end{figure}

The directions perpendicular to a lower dimensional stratum, see Figure \ref{fig:directions}, play a key role for the limiting behavior of Fr\'echet means in Theorem \ref{thm:stickyorth} below and the characterization of their topologies, as we develop in the sequel.

\begin{definition}
    
	\label{def:directionsSpaceAtStratum}

	Let $N \geq 3$, $\ort \subset \bhv$ be a stratum with positive codimension
	$l \geq 1$ and $T \in \ort$.  Then we have the directions at $T$\emph{perpendicular} and \emph{parallel} to $\ort$, respectively, given by
	\begin{align*}
		\perpdir           & := \{ \dir_T(T^\prime) \vert T^\prime \neq T, \ T^\prime \in
		\prj_{\overline{\ort}}^{-1}(\{T\}) \},                                           \\
		\Sigma_T^\parallel & := \{ \dir_T(T^\prime) \vert T^\prime \neq T, \
		T^\prime \in \ort \}.
	\end{align*}
\end{definition}

\subsection{The Fr\'echet Mean in BHV Spaces}

    Recall the definition of the Fr\'echet $F_\prb$ function from (\ref{eq:Frechetfcn}) which is well defined for $\prb \in  \wst[2]{\bhv}$, and the definition of the Fr\'echet mean set (\ref{eq:Frechetmean}), which is well defined for  $\prb \in  \wst[1]{\bhv}$ since $\bhv$ is complete and, as remarked in the introduction,  it is a unique point 
    $$ \mu = \argmin_{T \in
			\bhv} F_\prb(T)\,,$$
    because $\bhv$ is a Hadamard space. 
    
    Similarly, for a finite set $\mathcal{T}\subset \bhv$ of trees, we have the \emph{sample Fr\'echet function}
\begin{align*}
	F_{\mathcal{T}}(T) := \frac{1}{ 2 \vert\cT\vert}
	\sum_{T' \in \cT}^n d^2(T, T')\,, \quad T \in \bhv\,,
\end{align*}
    where  $\vert\cT\vert< \infty$ is the cardinality of $\cT$. Also, we have a unique sample mean
    $$\hat \mu = \argmin_{T \in
			\bhv} F_\cT(T)\,.$$

For $\prb \in \wst[2]{M}$, $T \in \bhv$ and a direction $\sigma \in \Sigma_T$, the \emph{directional derivative} of $F_\prb$ in direction $\sigma$ is given by
\begin{align*}
	\nabla_\sigma F_\prb(T) = \lim_{\lambda \searrow 0} \frac{ F_\prb(\gamma_T^\sigma(\lambda)) - F_\prb(T) }{\lambda}\,,
\end{align*}
where $\gamma_T^\sigma:[0,a] \to \bhv$ is the unit speed geodesic with $\gamma_T^\sigma(0) = T$ and direction $\sigma$ at $T$. There is an explicit representation of this derivative and a characterization of the Fr\'echet mean detailed below.


\begin{theorem}[{\cite[Theorems 4.4 and 4.7]{stickyflavs_arxiv}}]
\label{thm:dir-dertivative}

	\label{thm:optimality}
	Let $N\geq 3$, $T \in \bhv$ and $\prb \in \wst[2]{\bhv}$. Then, 
 
   \begin{itemize}
        \item[(i)] 
for every
	$\sigma \in \Sigma_T$, we have
	\begin{align*}
		\nabla_\sigma F_\prb(T) = - \int_{\bhv} d(T, T^\prime) \cos(\angle_T(\sigma, \dir_T(T^\prime))) \ \diff \prb(T^\prime)\,,
	\end{align*}
        \item[(ii)] the Fr\'echet
	mean of $\prb$ is given by a point $\mu \in \bhv$ if and only if
	\begin{align*}
		\nabla_\sigma F_\prb(\mean) \geq 0 \quad \forall \sigma \in \Sigma_\mean\,,
	\end{align*}
   \end{itemize}
   \end{theorem}

Unlike the Euclidean case, where equality holds in (ii)  of Theorem \ref{thm:optimality}, the inequality can be strict in some directions if $\mean$ does not have a fully resolved topology. Notably, it can be strict for all directions at the star tree, or orthogonal to a stratum. In these cases, the behavior of sample means may deviate from the Euclidean law of large numbers as detailed in Theorem \ref{thm:stickyorth}. Such a distribution is called \emph{sample sticky}, see \cite{bhvsticky, stickyflavs_arxiv} for this and other "flavors of stickiness". 

\begin{theorem}[{\cite[Theorem 3, Corollary 7]{barden18}}]
	\label{thm:stickyorth}

	Let $\ort \subset \bhv$ be a stratum of codimension $l \geq 1$ and $\prb
		\in \wst{\bhv}$ with $\mean = \mathfrak{b}(\prb) \in \ort$.
	For every $\sigma \in \Sigma^\perp_\mean$, it holds that $\nabla_\sigma F_\prb(\mean) \geq 0$.

	If further $\nabla_\sigma F_\prb(\mean) > 0$ for every
	$\sigma \in \Sigma_\mean^\perp$, then for every arbitrary sequence
	$X_1, X_2 \iid \prb$, there almost surely exists a random
	$N \in \NN$ such that
	\begin{align*}
		\smean_n \in \ort \quad \forall n \geq N\,.
	\end{align*}

\end{theorem}

 For arbitrary Hadamard spaces \cite{sturm} proposed an algorithm computing \emph{inductive means} that converge in probability to the Fr\'echet mean, where, under bounded support (e.g. for sample means), convergence is even a.s.: Starting with a first random tree $\hat \mu_0$, the  $j$-th \emph{inductive mean} $\hat \mu_j$, $j\in \NN$ is given traveling the geodesic, parametrized by the unit interval, from $\hat \mu_{j-1}$ to the $j$-th random tree $Y_j$ only until $\frac{1}{j+1}$. For our purposes here is its sample version.

\begin{algorithm}
    \caption{Sturm's algorithm.}\label{alg:sturm}
    \KwData{Trees $T_1, T_2, \ldots, T_n$
    }
    $j \gets 0$;\\
    Draw $Y_0\sim \frac{1}{n}\sum_{i = 1} \delta_{T_i}$ and set $\hat\mu_0  \gets Y_0$;\\
    \Repeat{\rm convergence}{
    $j \gets j + 1$;\\
      Draw $Y_{j} \sim \frac{1}{n}\sum_{i = 1} \delta_{T_i}$ independent of the $Y_0,\ldots,Y_{j-1}$;\\
      $\hat \mu_{j} \gets \overline{\gamma}_{\hat \mu_{j-1}}^{Y_{j}}\left(\frac{1}{j+1} \right)$;
       }
\end{algorithm}

If the Fr\'echet mean is located on a lower-dimensional stratum, with some trees featuring splits in higher-dimensional strata, compatible to some in the lower dimensional stratum, the output of Sturm's algorithm will, while metrically close, not necessarily have the correct unresolved topology. This behavior is explained in the following theorem. 

\begin{theorem}\label{thm:sturms-algo-not-singular}

	Let $N \geq 3$ and $\mathcal{T}  \subset \bhv$ be a finite set of trees,
	with its Fr\'echet mean $\mean \in \ort$ lying in a stratum $\ort$ of codimension $l \geq 1$.
	Let $\hat{\mu}_j$, $j \in \NN$, denote the output of Sturm's algorithm and suppose that 
    $$\{T \in \mathcal{T} : \exists s \in E(T)\mbox{ with } s\in   C(\mean) \setminus E(\mu) \} \neq \emptyset\,.$$    
    Then,	\begin{align*}
		\mathbb{P} \left( \hat{\mean}_j \notin \ort \text{ $\infty$-often} \right) = 1\,.
	\end{align*}
\end{theorem}

The goal of the following section will be, first, finding criteria identifying the Fr\'echet mean's correct topology, and then, develop algorithms for verification in practice.

\section{Finding the Topology of the Fr\'echet Mean}
\label{sec:mean_top}

In this section we first give our main result which at once leads to an algorithm to determine splits present in the topology of a Fr\'echet sample mean. While the main result is generally applicable to arbitrary strata, for the actual computation, only the corresponding result for the star tree stratum is of concern, and this rewrites in a simple form. The reason, why it suffices to consider star tree strata only is given in the second part where, among others the tangent cone at a stratum is decomposed into a product of the tangent space along the stratum and a product of lower dimensional BHV spaces -- of which only their star tree strata are of concern. Essentially, all hinges on directional derivatives orthogonal to the original stratum and for these, the offset inside the stratum is irrelevant.

 
\subsection{A Sufficient Condition for a Split in the Fr\'echet Mean}

While directional derivatives at the Fr\'echet mean yield a sufficient and necessary condition for a point to be the Fr\'echet mean, surprsingly, information about its topology can also be obtained from certain directional derivatives at the star tree. 
Recall the notation $\sigma_s = \dir_T(T + 1 \cdot s) \in \Sigma_T$
for $T \in \bhv$ and $s \in C(T)$.

\begin{theorem}
	\label{thm:main-strata}
	
	Let $N \geq 3$, $\ort\subset \bhv$ an orthant of codimension $0\leq l < N-2 $, $\prb \in \wst[2]{\bhv}$ with Fr\'echet mean 
    $\mean \not \in \ort$ and $T\in \ort$ so that $E(T)\subset E(\mean)$. Then the following hold:
    \begin{itemize}
        \item[(i)] 
if $s\in C(T)\setminus E(T)$ with $\angle_T(\mean, \sigma_s)) \geq \pi /2$, 
    then $\nabla_{\sigma_s} F_\prb(T) \geq 0$,
    \item[(ii)] if $\nabla_{\sigma_s} F_\prb(T) < 0$, then $s \in E(\mean)$.
    \end{itemize}
    \end{theorem}

%
%

\begin{remark}
    \label{rem:const_dir}

    By Lemma \ref{lemma:const_deriv} in the following section, the directional derivatives do not depend on the 
    particular choice of the point $T \in \ort$, only on its topology.
\end{remark}

Theorem \ref{thm:main-strata} entails the following algorithm for finding splits present in the Fr\'echet mean's topology. 

\begin{algorithm}[!ht]
    \caption{An algorithm for finding splits in the Fr\'echet mean.}
	\KwData{trees $\mathcal{T}= \{T_1, \ldots, T_n\}$}
    $\tilde{\mean} \gets \stree$;\\
    $\cE \gets \bigcup_{T \in \cT} E(T) $;\\
    \Repeat{$E(\tilde\mean_\textrm{old}) = E(\tilde\mean)$}{
        $\tilde{\mean}_\mathrm{old} \gets \tilde{\mean}$;\\
    	\For{$s \in \cE$}{
            \If{$\nabla_{\sigma_s} F_\cT(\tilde\mean_\mathrm{old})< 0$}{$\tilde{\mean} \gets \tilde{\mean} + 1\cdot s$;}
    	}
        $\cE \gets \cE \cap \left(C(\tilde{\mean}) \setminus E(\tilde{\mean})\right)$;\\
    }
    \label{alg:find-splits}
\end{algorithm}

Applying Theorem \ref{thm:main-strata} to $\ort = \{\stree\}$ and sample means yields the following. 


\begin{corollary}
	\label{cor:sample}

	Let $\mathcal{T} 
    \subset \bhv$ be a finite collection of trees
	with sample Fr\'echet mean $\smean$.
	Then, $s \in E(\smean)$ if
	\begin{align}
		\label{eq:new_cond}
		\sum_{T \in \mathcal{T}}
        \sqrt{\sum_{C(s)\not \ni x \in E(T)} \lvert x \rvert_T^2}
		< \sum_{T \in \mathcal{T}} \lvert s \rvert_T \;.
	\end{align}
\end{corollary}

This result improves a weaker condition from \cite[Theorem 5]{bhvmean} stating
that a split
$s \in \cup_{T\in \mathcal{T}}E(T)$ is contained in the sample Fr\'echet mean if
\begin{align}
	\label{eq:old_cond}
	\sum_{T \in \mathcal{T}} \sum_{C(s)\not \ni x \in E(T)}
	\lvert x \rvert_T < \sum_{T \in \mathcal{T}} \lvert s \rvert_T\,.
\end{align}
We illustrate the two conditions (\ref{eq:new_cond}) and (\ref{eq:old_cond}) by example from \cite[Example 1 and 2]{bhvmean}. While our condition is stronger, it is still not necessary. 

\begin{example}
	\label{ex:improv}

    \cite[Example 1]{bhvmean} consider the sample mean $\mu$ of four trees $\mathcal{T} = \{T_1, T_2, T_3, T_4\} \subset \mathbb{T}_4$ as shown in Figure \ref{fig:example} where for $T_1$ the split $s_1$ has length $w>0$. As there is only one  tree, namely $T_4\in \mathcal{T}$ having splits  not in $C(s_1)$,  (\ref{eq:old_cond}) is equivalent with
    $$ 10 + 10 = \sum_{x\in E(T_4)} \lvert x\rvert_{T_4} < \sum_{T\in \mathcal{T}} \lvert s_1\rvert_{T} = w + 3 + 1\,,$$
    i.e. $s_1\in \mu$ if $w > 16$.

    In contrast, (\ref{eq:new_cond}) rewrites to 
    $$  \sqrt{10^2 + 10^2} = \sqrt{\sum_{x\in E(T_4)} \lvert x\rvert^2_{T_4}} < \sum_{T\in \mathcal{T}} \lvert s_1\rvert_{T} = w + 3 + 1\,,$$
    and we obtain the stronger result that $s_1\in \mu$ if $w > \sqrt{2} \cdot 10 - 4 \approx 10.142$.

    This is, however, not sharp, since  direct computation by  \cite[Example 1]{bhvmean} showed that $s_1\in \mu$ if $w > 9.82$ ( below that and until $9.51$, $\mu = \star$).


	\begin{figure}[!h]
		\centering
		\subfloat[\it Level curves of the Fr\'echet function and part of the link (red, for better visibility enlarged by a factor of 5).]{\includegraphics[width=7cm]{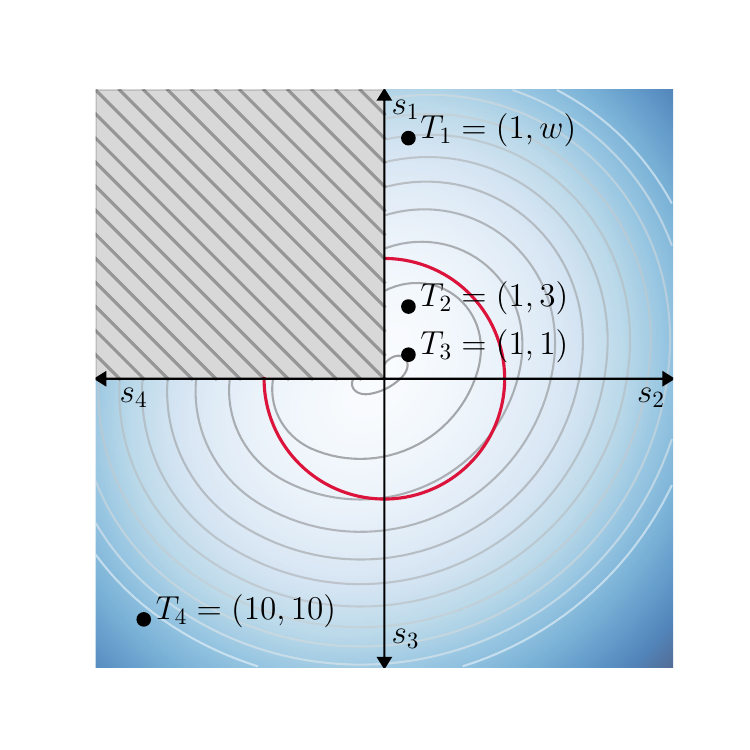}}\qquad
		\subfloat[\it Values of the directional derivatives of the Fr\'echet function along directions to the part of the link (red) 
    displayed on the left hand panel.]{\includegraphics[width=7cm]{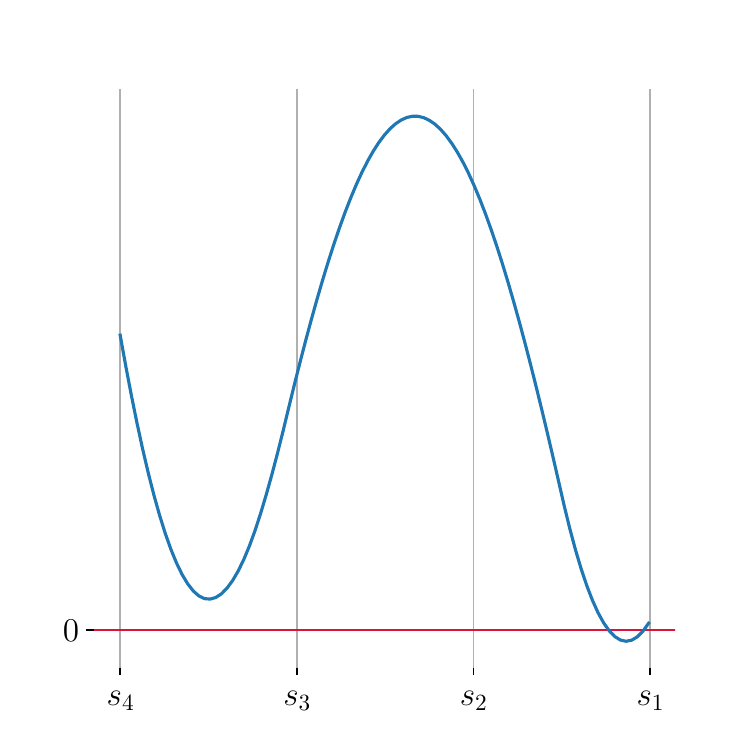}}\qquad

		\caption{\it Illustrating the Fr\'echet function and its directional derivatives of Example \ref{ex:improv} in the subspace 			inhabited by the trees $T_1,T_2,T_3,T_4 \in \mathbb{T}_4$ for the choice $w=10$.
  }
		\label{fig:example}
	\end{figure}

\end{example}

\subsection{Tangent Cone and Space of Directions at Lower Dimensional Strata}

Let $\ort \subset \bhv$ be a stratum with positive codimension $l \geq 1$. As was mentioned in \cite{barden18}, the spaces of directions of points in $\ort$ inherits the nature of a stratified space from $\bhv$. Theorem \ref{thm:tancone}, a more explicit version of \cite[Lemma 2]{bhvsticky}, shows how both the spaces of directions and the tangent cone decompose. To this end, introduce the following notion.

\begin{definition}
    \label{def:joint}
    
    For two metric spaces $(M_1, d_1), (M_2, d_2)$, their \emph{spherical join} is
    given by
    \begin{align*}
    	M_1  \ast M_2 = \left[0,\frac{\pi}{2}\right]\times M_1\times M_2/\sim \ \cong
    	\left\{(\cos \eta\, p_1,\sin \eta\, p_2): 0\leq \eta \leq
    	\frac{\pi}{2}, p_i\in M_i,i=1,2\right\}\,,
    \end{align*}
    i.e. the equivalence class of $(\eta,p_1,p_2)$ contains the single point $(\eta,p_1,p_2)$ unless $\eta =0,\pi/2$. For $\eta =0$, the class contains all of $M_2$ in the last component, hence the class is uniquely determined by $p_1\in M_1$, and for $\eta = \pi/2$, it contains all of $M_1$ in the second component, hence the class is uniquely determined by $p_2\in M_2$.  The spherical join is 
    equipped with the metric
    \begin{align*}
    	d\big((\eta,p_1,p_2),(\eta^\prime,p_1^\prime,p_2^\prime)\big) =
    	\arccos\left(\cos\eta\cos\eta' \cos(d_1(p_1,p'_1)) + \sin\eta
    	\sin\eta' \cos(d_2(p_2,p'_2))\right)\,,
    \end{align*}
    see e.g. \cite[Definition 5.13]{bridson}.
    
    For metric space $(M_1, d_1), \ldots, (M_k, d_k)$, $k\geq 3$
    define by induction the \emph{nested} spherical join
    \[
    M_1 \ast \Bigg( M_2 \ast \Big(M_3 \ast \big(\cdots  \ast( M_{k-1} \ast M_k)\cdots \big)\Big) \Bigg) = \left[0,\frac{\pi}{2}\right]^{k-1}\times M_1\times \ldots \times M_k/\sim\,.
    \]
    Points, i.e. equivalence classes are determined by their coordinates $(\eta_1, \ldots, \eta_{k-1}, p_1, \ldots, p_k)$ with $\eta_j \in [0,\pi/2]$, $j=1, \ldots, k-1$, and $p_j \in M_j$, $j=1,\ldots, k$ and coordinates uniquely determine their points if 
    $\eta_j \in (0,\pi/2)$ for all $j=1, \ldots, k-1$. Moreover, 
    $M_1$ will be identified with all coordinates having $\eta_1 =0$, and for $i=2,\ldots,k$, $M_i$ will be identified with all coordinates having  $\eta_j =\pi/2$ for $1\leq j \leq i-1$ and $\eta_i =0$. 
\end{definition}

    \begin{remark}
    It is easy to see that the 
     distance between points with coordinates
    \[p = (\eta_1, \ldots, \eta_{k-1}, p_1, \ldots, p_k), \quad p^\prime = (\eta_1', \ldots, \eta_{k-1}', p_1', \ldots, p_k')\] is then given by
    \begin{align}\label{eq:join-distance} \nonumber
        d\left(p, p^\prime\right) &= \arccos\Bigg(\sum_{i=1}^{k-1}\ \cos(\eta_i)\cos(\eta_i')\left(\prod_{j=1}^{i-1} \sin(\eta_j)\sin(\eta_j)^\prime \right) d_i(p_i, p_i^\prime) \\
        & \ + \prod_{i=1}^{k-1} \sin(\eta_i)\sin(\eta_i)^\prime d_k(p_k, p_k^\prime)\Bigg)\,.
    \end{align}
    \end{remark}

\begin{theorem}[Decomposition theorem]
	\label{thm:tancone}

	Let $N \geq 3$, $\ort \subset \bhv$ be a stratum with positive codimension $l \geq 1$ and $T \in \ort$.  Then the following hold:
 \begin{itemize}
     \item[(i)] 
 there are $2\geq m\in \NN$ and BHV spaces $\TT_{k_j}$ of dimensions $1\leq k_j$, $j=1,\ldots,m$ with \mbox{$\sum_{j=1}^m (k_j - 2) = l$} such that the tangent cone and the space of directions at $T$ decompose as follows:
	\begin{eqnarray*}
		\mathfrak{T}_T &\cong& \mathbb{R}^{N-l-2} \times \mathbb{T}_{k_1} \times \cdots \times \mathbb{T}_{k_m}\\
  \Sigma_T &\cong& \Sigma_T^\parallel \ast
		\perpdir \quad\mbox{ with }\quad \Sigma_T^\parallel ~\cong~ S^{N-l-3} \quad \mbox{ and }
  \\		\perpdir &\cong& \mathbb{L}_{k_1} \ast \Bigg( \mathbb{L}_{k_2} \ast \Big(\mathbb{L}_{k_3} \ast \big(\cdots  \ast( \mathbb{L}_{k_{m-1}} \ast \mathbb{L}_{k_m})\cdots \big)\Big) \Bigg)
   \,,
	\end{eqnarray*}
	where the links are equipped with the Alexandrov angle as metric;
     \item[(ii)] 
     with the notation from (i), the canonical projections    
 $$\varpi_j: \mathfrak{T}_T \to \mathbb{T}_{k_j},\quad j=1,\ldots,m\,,$$
 constructed in the proof of Lemma \ref{lemma:shared_edges} in the appendix,
and $ \sigma = (\eta_1, \ldots, \eta_{m-1},\sigma_1, \ldots \sigma_m ) \in \perpdir$ we have for arbitrary $T'\in \bhv$
%
	that
	\begin{align*}
		d(T, T^\prime) \cdot \cos(\angle_T(\sigma, T^\prime)) = & \sum_{i=1}^{m-1} \left(\prod_{j=1}^{i-1} \sin(\eta_j)\right)\cos(\eta_i) \cdot d(\stree, \varpi_i(\log_T(T^\prime))) \cdot \cos(\angle_{\stree}(\sigma_i, \varpi_i(\log_T(T^\prime)))) \\
		                                                        & + \left(\prod_{j=1}^{m-1} \sin(\eta_j)\right) \cdot d(\stree, \varpi_m(\log_T(T^\prime))) \cdot \cos(\angle_{\stree}(\sigma_m, \varpi_m(\log_T(T^\prime)))) \,.
	\end{align*}

 \end{itemize}
 
 
\end{theorem}

Immediately, we obtain the following corollary.

\begin{corollary}
	\label{cor:perpderivs}

	Let $N\geq 3$, $\ort \subset \bhv$ a stratum of codimension $l \geq 1$,
	$T \in \ort$ and $\prb \in \wst[2]{\bhv}$.
	Suppose $\mathfrak{T}_T \cong \RR^{N-l-2} \times \mathbb{T}_{k_1} \times \ldots \times \mathbb{T}_{k_m}$.
	Then,
	\begin{align*}
		\nabla_\sigma F_\prb(T) \geq 0 \quad \forall \sigma \in \perpdir
	\end{align*}
	if and only if for every $i\in \{1, \ldots, m\}$
	\begin{align*}
		\nabla_\sigma F_\prb(T) \geq 0 \quad \forall \sigma \in \mathbb{L}_{k_i} \subset \perpdir \,.
	\end{align*}

\end{corollary}

Since the directions in $\perpdir$ correspond to the addition of splits that
are compatible with the topology of a stratum $\ort$, one can naturally 
identify $\Sigma_{T_1}$ with $\Sigma_{T_2}$ for $T_1, T_2 \in \ort$. 
As the following lemma teaches, the perpendicular directional derivates 
of the Fr\'echet function do not depend on the particular reference point
in the orthant.

\begin{lemma}[{\cite[Corollary 1]{bhvsticky}}]
    \label{lemma:const_deriv}

    Let $N \geq 3$ and $\ort$ be a stratum of codimension $l \geq 1$. Then, 
    for any $T_1, T_2 \in \ort$ and $\sigma \in \Sigma_{T_1} \cong \Sigma_{T_2}$ that
    \[
        \nabla_\sigma F_\prb(T_1) = \nabla_\sigma F_\prb(T_2)\,.
    \]
    
\end{lemma}

\section{Minimizing the Directional Derivative}
\label{sec:min_dir_deriv}

In view of  Theorem \ref{thm:dir-dertivative}, in order to see that a tree $\mean$ is a Fr\'echet mean of a sample $\cT$, we minimize $\nabla_\sigma F_{\cT}(\mu)$ over $\sigma \in \Sigma_\mu$ and show that this minimum is nonnegative. Finding the minimizer, however, requires solving an optimization problem on $\Sigma_\mu$. Due to decomposition in Theorem \ref{thm:tancone} it suffices to obtain minima on the star stratum only.
The space $\Sigma_\stree$ inherits the non-smooth nature of the tree space $\bhv$. As the number of orthants is given by $(2N -3)!!$, searching for maxima for the direction corresponding to fixed orthants quickly becomes infeasible for larger leaf set sizes.

Deriving a proximal splitting algorithm algorithm and a stochastic gradient algorithm to determining the minimum, instead of traversing all of these many orthants, is the subject of this section. 

For a sample $\cT = \{T_1, \ldots T_n\}, \cT\subset  \mathbb{T}_N$ due to Theorem \ref{thm:dir-dertivative}, we want to minimize 
$$f(\sigma) := \frac{n^2 \, \nabla_\sigma F_\cT(\stree) }{\sum_{j=1}^{N} d(\stree, T_j)}= \sum_{i=1}^n w_i f_{T_i}(\sigma),\quad \sigma \in \Sigma_\stree$$
where 
$$	f_T(\sigma) = -  \cos(\angle_\stree(\sigma, T)),\quad w_i = \frac{d(\stree, T_i)}{\frac{1}{n}\sum_{j=1}^{N} d(\stree, T_j)}, \quad i=1,\ldots,n\,. $$
For fixed $0 < \nu \leq 1/w_i$ for all $i=1,\ldots,n$ (the reason becomes clear in the proof of the Lemma \ref{lem:prox}below), define the \emph{proximal splitting operator} 
$$ \prox_T(\tau) := \argmin_{\sigma \in \Sigma_\stree} \left\{ f_T(\sigma) + \frac{1}{2 \nu}\angle_\stree(\tau, \sigma)^2 \right\}\,,\quad \tau \in \Sigma_\stree\,.$$

The following lemma teaches that $\prox_T(\tau)$ is unique and that its computation is rather simple. To this end, recall that $\beta_{\tau}^{{\dir_\stree(T)}}$ is the geodesic in $\Sigma_\stree$ mapping $0$ to $\tau$ and $1$ to $T$, while $\bar \gamma_{T}^{T'}$ is the geodesic in $\TT_N$ mapping $0$ to $T$ and $1$ to $T'$.

\begin{lemma}\label{lem:prox}
For $N \geq 4$, $\tau \in \Sigma_{\stree}$, $T\in  \TT \setminus \{\stree\}$ with $\angle_{\stree}(\dir_\stree(T), \tau) \in [0,\pi]$ we have
\begin{enumerate}
    \item[(i)] 
$\prox_T(\tau) = \left\{\bar \beta_{\tau}^{{\dir_\stree(T)}}(\lambda)\right\}$
where 
$$\lambda = \argmin_{0\leq \lambda' \leq  \angle_{\stree}(\dir_\stree(T), \tau)} \left(\bar \beta_{\tau}^{{\dir_\stree(T)}}(\lambda') + \frac{\lambda^2}{2\nu}\right)\,, $$
and, in particular
$\lambda =0$ in case of $\angle_{\stree}(\dir_\stree(T), \tau) \in\{0,\pi\}$, i.e. $\prox_T(\tau) = \left\{\tau \right\}$ then;
\item[(ii)] in case of $\angle_{\stree}(\dir_\stree(T), \tau) < \pi$,
$$\bar{\beta}_{\tau}^{\dir_\stree(T)}(\lambda) = \dir_\stree \left(\overline{\gamma}_{T_\tau}^{T/\|T\|}\left( \lambda^\prime \right) \right), \quad \lambda \in [0,1] $$
where $T_\tau \in  \TT_N$ with $\dir_\stree(T_\tau) = \tau$ and $\|T_\tau\| = 1$, and 
	\begin{align*}
  \lambda^\prime= \frac{\sin(\lambda \cdot \angle_\stree(\sigma, \tau))}{2 \cdot \sin(\angle_\stree(\sigma, \tau)) \cdot \sin\left( \frac{\pi + (1-\lambda) \cdot \angle_\stree(\sigma, \tau)}{2}\right)}\,, \quad \lambda \in [0,1].
	\end{align*}
\end{enumerate}
\end{lemma}

We denote by by $p_T(\tau) \in \prox_T(\tau)$ the unique element for $T\in \TT_N, \tau \in \Sigma_\stree$.
The functions $f_{T_i}: \{\tau \in \Sigma_\stree: \angle_{\stree}(\tau,\sigma) < \pi/4\} \to \RR$ are convex as the negative cosine is thus on $[0,pi/2]$,  hence we can apply \cite[Theorem 4]{prox_split_cat} yielding the following assertion for iterates of the corresponding proximal operator on CAT(1) spaces. 

\begin{theorem} 
    \label{thm:lausterluke}
    
    For $N \geq 4$, let $\{T_1, \ldots, T_n \}\subset \bhv$ and assume there is a direction $\sigma \in \Sigma_\stree$
    such that $\dir_\stree(T_i) \in B_{\pi/4}(\sigma):= \{\tau \in \Sigma_\stree: \angle_\stree(\tau,\sigma) < \pi/4\}$ for all $i = 1,\ldots,n$, that 
    \begin{align*}
        p: &B_{\pi/4}(\sigma) \to B_{\pi/4}(\sigma),\quad 
        \tau \mapsto \Big(
         p_{T_n}\circ \ldots \circ p_{T_1} 
        \Big)(\tau)
    \end{align*}
    has a nonempty set of fixed points $\mathrm{Fix}(p) $, and that there is a constant $c >0$ such that
    $\angle_\stree(\tau, \mathrm{Fix}(p)) < c \cdot \angle_\stree(\tau,p(\tau))$ for $\tau \in B_{\pi/4}$.
    Then, for $\tau_0 \in B_{\pi/4}(\sigma)$ sufficiently close to $\mathrm{Fix}(p)$, the sequence 
    \[ 
        \tau_{m} := \Big(p_{T_n}\circ \ldots \circ p_{T_1}\Big)(\tau_{m-1})\,,\quad m\in \NN\,,
    \]
    converges to a point in $\mathrm{Fix}(p)$.
\end{theorem}

Theorem \ref{thm:lausterluke} as well as Lemma \ref{lem:prox} motivate the following algorithm.

\begin{algorithm}
	\caption{ A proximal splitting algorithm for finding local minima
		of $\sigma \mapsto \nabla_\sigma F_\mathcal{T}(\stree)$.}\label{alg:prox}
	\KwData{trees $\mathcal{T}= \{T_1, \ldots, T_n\}$, initial guess $\sigma_0$, 
 $0 < \nu \leq 1$ }
    $j \gets 0$;\\
    \Repeat{\rm convergence}{
		$\sigma_{j+1} \gets \prox_{f_n, \nu} \circ \prox_{f_{n-1}, \nu} \circ \ldots \circ \prox_{f_1, \nu}(\sigma_j)$;\\
        $j \gets j +1 $
	}
\end{algorithm}

Alternatively, one might use Algorithm \ref{alg:mindegrees}, inspired by
geodesic gradient descent methods. The randomized approach of such an algorithm
might be better suited for finding more local minima by `exploring' the space.
We do not, however, provide a proof of convergence for this algorithm.
The approach here is that, in each iteration, the current position is updated
by shooting a geodesic from the current position to a randomly drawn direction from the
data set that is less than $\pi$ away.

\begin{algorithm}
	\caption{An algorithm for finding local minima of $\sigma \mapsto \nabla_\sigma F_\mathcal{T}(\stree)$,
		inspired by stochastic gradient descent.}\label{alg:mindegrees}
	\KwData{trees $\mathcal{T}= \{T_1, \ldots, T_n\}$, initial guess $\sigma_0$
 }
    $j \gets 0$;\\
    \Repeat{\rm convergence}{
	$I_j \gets \{ i \in \{1,2, \ldots,n\} : \angle_\stree (\sigma_j, T_i) < \pi \}$;\\
	Draw $\tau_{j+1} \sim \sum_{i \in I_j} \frac{d(\stree, T_i)}{\sum_{l \in I_j} d(\stree, T_l)} \delta_{\dir_\stree{T_i}}$;\\
	$\sigma_{j+1} \gets \overline{\beta}_{\sigma_j}^{\tau_{j+1}}\left(\frac{\sin(\angle_\stree(\sigma_j, \tau_{j+1}))}{j+1} 
    \right)$;\\
    $j \gets j + 1$:\\
	}
\end{algorithm}

\begin{remark}

	By design, Algorithms \ref{alg:prox} and \ref{alg:mindegrees} are only
	capable of finding local minima. This is easy to see if the initial
	guess is `isolated'. In each iteration, the current position is
	updated by going towards a direction from the data set that is less than
	$\pi$ away. Going back to Example \ref{ex:improv}, starting at $\tau_0 =
		\dir_\stree(T_4)$, the algorithms would remain stationary at
	the local minimizer $\dir_\stree(T_4)$, regardless of the choice of the
	parameter $w \geq 1$.
 
	Thus, we recommend using both algorithms multiple times with different
	initial guesses.

\end{remark}

\begin{remark}
    \label{rem:whichalg}

    In terms of performance, Algorithm \ref{alg:prox} requires fewer computations 
    of angles than Algorithm \ref{alg:mindegrees}. Determining $I_j$ in Algorithm 
    \ref{alg:mindegrees} is the major computational bottleneck -- it requires the
    computation of all angles between the current position and the directions 
    from the data set. This can be alleviated by keeping $I_j$ for multiple  
    iterations before updating at the cost of precision.
    
    On the other hand, different runs of Algorithm \ref{alg:mindegrees} 
    seems to converge against different local minima, whereas Algorithm 
    \ref{alg:prox} tends to converge against the same local minimum, 
    see also Figure \ref{fig:api_alg} in Section~\ref{sec:api}.
    
    Since the link $\mathbb{L}_N$ stretches across all $(2N-3)!!$ orthants, it can
    be advisable to perform multiple runs of Algorithm \ref{alg:mindegrees} 
    despite the higher computational cost.
\end{remark}

\section{Hypothesis Testing in the Presence of Stickiness}
\label{sec:testing}
\subsection{Testing for the Presence of Splits in the Fr\'echet Mean}

Building on Theorem~\ref{thm:main-strata}
we derive the following one-sample test for the hypotheses
\begin{align}\label{hypo:one-sample-test}
	\mathcal{H}_0: s \not\in E(\mean)\text{ vs. }
	\mathcal{H}_1 : s \in E(\mean)\,,
\end{align}
for the presence of a split $s$ in the population Fr\'echet mean $\mean$ of  probability distribution $\prb \in \wst[2]{\bhv}$.

\begin{test}[For the presence of a split in the Fr\'echet mean]\label{test:one-sample}
     Given a sample $T_1,\ldots,T_n \iid \prb$ and a level $0<\alpha < 1$, reject $\cH_0$ if
\begin{eqnarray*}
\hat Z := \frac{\sqrt{n}\,\bar X_n}{\widehat{\rm sdev}} < c_{\alpha}\mbox{ where }\bar X_n:= \frac{1}{n}\sum_{i=1}^n X_i,~\widehat{\rm sdev} := \sqrt{\frac{1}{n-1}\sum_{i=1}^n (X_i-\bar X_n)^2}\,.
\end{eqnarray*}
  Here, $$ X_i :=  -\cos(\angle_\stree(\sigma_s, T_i)) \cdot d(\stree, T_i)\,,\quad i \in \{1,2,\ldots,n\}\,,$$
  and $c_{\alpha}$ can be taken as the $\alpha$-quantile of the student $t_{n-1}$-distribution with $n-1$ degrees of freedom, i.e. $\mathbb P\{ Z \leq c_{\alpha}\} = \alpha$ for $Z\sim t_{n-1}$.
  
  Alternatively, $c_{\alpha}$ can be simulated by bootstrap sampling from the data centered by its sample mean (thereby simulating $\cH_0$): For $B\in \NN$ large (typically $1,000$) and each $b=1,\ldots,B$, let $n_b \in \NN$ ($=n$ in $n$-out-of-$n$-bootstrap) sample $X_1^*,\ldots,X_{n_b}^* \iid \frac{1}{n} \sum_{i=1}^n \delta_{X_i-\bar X_n}$ to obtain
  $$ Z^*_b := \frac{\sqrt{n_b}\, \bar X^*_b}{\sqrt{ \frac{1}{n_b-1}\sum_{i=1}^{n_b}(X^*_i-\bar X^*_b)^2}},\mbox{ where } \bar X^*_b := \frac{1}{n_b}\sum_{i=1}^{n_b}X^*_i\,,$$
  and
    \[
    	c_{\alpha} = \max\left\{x \in \RR: \frac{1}{B} \sum_{b=1}^{B} 1_{\left(-\infty,x\right]}\left(Z^*_b\right) \leq \alpha\right\}\,.
    \]

\end{test}

\begin{remark}
    \label{rem:one-sample}
    
    1) Since 
    $$-\cos(\angle_\stree(\sigma_s, T)) \cdot d(\stree, T) = \frac{1}{2}\,\nabla_s \,d(\stree,T)^2\,,$$
    for every $T \in \bhv$, cf. Theorem \ref{thm:dir-dertivative} and thus $\EE[X_i] = \nabla_s F_{\prb}(\stree)$ for all $1\leq n$, Test \ref{test:one-sample} is a classical student $t$-test for 
\begin{align}\label{hypo:one-sample-test-prime}
	\mathcal{H}'_0: \nabla_s F_\prb(\stree) \geq 0 \text{ vs. }
	\mathcal{H}'_1 : \nabla_s F_\prb(\stree) < 0,
\end{align}
    which is robust under nonnormality, i.e. keeping asymptotically the level $\alpha$ (e.g. \citet[Section 11.3]{RomanoLehmann2005}).
    Although $\nabla_s F_{\prb}(\stree)< 0$ implies $s\in E(\mu)$, due to Theorem \ref{thm:main-strata},$ \cH'_0$ may be true without $ \cH_0$ being true, cf. Example \ref{ex:improv}, the true level of Test \ref{test:one-sample} for  (\ref{hypo:one-sample-test}) may be higher than its nominal level, making it more liberal.

    2) The necessary condition of \cite{bhvmean}, see Inequality \eqref{eq:old_cond}, was derived for the Fr\'echet mean of a finite set of trees. Our improved condition in Theorem~\ref{thm:main-strata}, see Inequality \eqref{eq:new_cond}, however, is applicable to general probability distributions 
$\prb \in \wst[2]{\bhv}$.
\end{remark}

We can adapt this test to a more general setting \emph{if some splits of the population mean are already known}.

\begin{test}[For the presence of a split in the Fr\'echet mean if other splits are known]\label{test:one-sample-strata}
   Given a sample $T_1,\ldots,T_n \iid \prb$ and a level $0<\alpha < 1$, assume we have
   have a stratum $\ort \subset \bhv$ such that for every $T^\prime \in \ort$, one has 
   $E(T^\prime) \subset E(\mean)$, where $\mean$ is the unique Fr\'echet mean of
   $\prb$. 
   Further, suppose $s \in C(T^\prime)$ and $\mathfrak{T}_{T^\prime} \cong \RR^{N-l-2} \times \mathbb{T}_{k_1} 
   \times \ldots \times \mathbb{T}_{k_m}$ as in Theorem \ref{thm:tancone}.
   Let $\varpi_r : \mathfrak{T}_T \to \TT_{k_r}, r\in \{1,\ldots,m\}$,
   be the canonical projection such that $\varpi_r(T^\prime + 1 \cdot s) \neq \stree$.
   By projecting, we obtain the sample $\cT^\prime = \{{T'_1},\ldots, {T'_{n}}\}$, where 
     \begin{align*}
        T_i^\prime &= \varpi_r(\log_\mean(T_i))\,, \qquad i=1,\ldots,n\,.
     \end{align*}
   Reject $\mathcal{H}_0$ from (\ref{hypo:one-sample-test}) if Test \ref{test:one-sample} rejects $\cH'_0$ (\ref{hypo:one-sample-test-prime})for $\cT^\prime$ and 
   $\sigma_s = \dir_\stree(\varpi_r(T^\prime + s))$.
\end{test}

\subsection{A Two Sample Test for Distributions with Same and Sticky Mean}

Building on the central limit theorem for directional derivatives from \cite[Theorem 6.1]{stickyflavs_arxiv}, see also \cite[Theorem 2]{mattingly2024central_tanfields}, for two probability distributions $\prb,\prbalt\in \wst[2]{\bhv} $ with common Fr\'echet mean $\mu$ on a stratum $\mathbb S$ of codimension $1 \leq l \leq N-2$ we propose the following two-sample test for the hypotheses
\begin{align}\label{hypo:two-sample-test}
	\mathcal{H}_0:  \prb = \prb^\prime\text{ vs. }
	\mathcal{H}_1 : \prb \neq \prb^\prime\,,
\end{align}
for the equality of $\prb$ and $\prb^\prime$, which is motivated by the procedure described in \cite[Section 3.7.1]{vdvaart}.

The second test below treats the general case by reducing it via Theorem \ref{thm:tancone} to the special case of $\ort = \{\stree\}$, treated by the first test below.

\begin{test}[For equality in the presence of stickiness to $\stree$]\label{test:two-sample-star}
    Let $\Sigma \subseteq \Sigma_\stree$.

     Given two independent samples $\cT=\{T_1,\ldots,T_n\}, T_i \iid \prb$ ($i=1,\ldots,n$) and $\cT^\prime=\{T^\prime_1,\ldots,T^\prime_{n'}\}$, $T^\prime_j \iid \prb^\prime$ ($j=1,\ldots,n$'), and a level $0<\alpha < 1$,
     let 
    \begin{align*}
        X_{i,\sigma} :=  -\cos(\angle_\stree(\sigma, T_i)) \cdot d(\stree, T_i)\,,\quad i \in \{1,\ldots,n\}\,,\\
        X^\prime_{j,\sigma} :=  -\cos(\angle_\stree(\sigma, T^\prime_i)) \cdot d(\stree, T^\prime_j)\,,\quad j \in \{1,\ldots,n'\}\,,  
    \end{align*}
     Then reject $\cH_0$ if
    \begin{eqnarray*}
        Z_{(\cT, \cT^\prime)}:=\sup_{\sigma \in \Sigma } \lvert \bar X_{\sigma} - \bar X^\prime_{\sigma}\rvert > c_{1-\alpha} \mbox{ where } \bar X_{\sigma}:= \frac{1}{n}\sum_{i=1}^n X_{i,\sigma}, \ \bar X^\prime_{\sigma}:=\frac{1}{n'}\sum_{j=1}^{n'}X^\prime_{j,\sigma}\,.
    \end{eqnarray*}
    Here, $c_{1-\alpha}$ is obtained through permutation of samples.
    By sampling from $\cT \cup \cT'$ without replacement, we generate pairs
   $(\cT^1, {\cT'}^1),\ldots, (\cT^B, {\cT'}^B)$, for large $B\in \NN$ (typically $1,000$)
    with 
    \[  
        \lvert \mathcal{T}^b\rvert = n \text{ and } \lvert {\mathcal{T}'}^b\rvert =n'\,, \qquad b=1,\ldots, B\,.
    \]
    For these permuted samples, we evaluate the  statistics 
    $Z_{(\mathcal{T}^1, {\mathcal{T}'}^1)}, \ldots, 
    Z_{(\mathcal{T}^B, {\mathcal{T}'}^B)} $ and determine 
    \[
    	c_{\alpha} = \max\left\{x \in \RR: \frac{1}{B+1} \left(1 + \sum_{b=1}^{B} 1_{\left(-\infty,x\right]}\left(Z_{(\mathcal{T}^b, {\mathcal{T}'}^b)}\right)\right) \leq \alpha\right\}\,.
    \]
    Here, in addition to \cite[Section 3.7.1]{vdvaart} we have added  one to the sum in order to avoid p-values of zero (see \cite{permutation_zero}).
%
%
    Then, the p-value is given by 
    \[
        \mathrm{p} = 
        \frac{1}{B+1} \left(1 + \sum_{b=1}^{B} 1_{\left(-\infty,x\right]}\left(Z_{(\mathcal{T}^b, {\mathcal{T}'}^b)}\right)\right)\,.
    \]
\end{test}

Ideally, one would choose $\Sigma = \Sigma_\stree$ in Test~\ref{test:two-sample-star}.
In the absence of suitable numerical methods, this approach quickly becomes
computationally infeasible with larger $N$, as the space of
directions at the star tree corresponds to a sphere stretching across all
$(2N-3)!!$ orthants. Instead, we propose using the following finite 
selections of directions.

\begin{definition}
	Let $N \geq 3$ and $\cT, \cT'\subset \bhv$ be two finite subsets of trees. Then let
	\begin{align*}
		\Sigma_1 & = \{ \dir_\star(T) \ : \ T \in \cT \cup \cT' \}\,, \text{ and}                     \\
		\Sigma_2 & = \{ \dir_\star(\star + 1 \cdot s) \ : \ \exists T \in \mathcal{T} \cup \mathcal{T}'\mbox{ with }s \in E(T) \} \,.
	\end{align*}
\end{definition}

\begin{remark}[Computational complexity]
\label{rem:comp_twosample}

Let us briefly discuss the computational complexity of Test \ref{test:two-sample-star}
for both choices of directions.
Given a direction $\sigma \in \Sigma_1\cup \Sigma_2$ (of course this also holds for all $\sigma \in \Sigma_\stree$), we have that
\[
\lvert \bar X_{\sigma} - \bar X^\prime_{\sigma}\rvert = \left\lvert \frac{1}{n} \sum_{i=1}^{n} \lVert  T_i \rVert \cos(\angle_\stree(\sigma, T_i))- \frac{1}{n'} \sum_{j=1}^{n'} \lVert T'_j \rVert \cos(\angle_\stree(\sigma, T'_j))\right\rvert\,.
\]
Here, the main computational burden lies in computing the angles, which can be done via the Euclidean law of cosines for
a tree representing a direction (see Proposition \ref{prop:bhv_cone}).
For $\Sigma_2$ we need to computate of all angles of the type $\angle_\stree(\stree + 1 s, T)$, $T \in \cT \cup \cT'$. 
By Remark \ref{rem:dist_single_split}, this requires determining the pairwise compatibility of all splits.
As each tree has at most $N-2$ splits, the number of splits is bounded from above by $(n + n') \cdot (N-2)$.
Recall verifying the compatibility $A_1 \vert A_2$ and $B_1 \vert B_2$ is done
by computing the intersections $A_i \cap B_j$, $i,j \in \{1,2\}$. 
Thus, determining the compatibility of two splits is of complexity $\mathcal{O}(N)$.
In total, we obtain a complexity of $\mathcal{O}(N^3 \cdot (n_x + n_y)^2)$ for computing
the pairwise compatibility of the splits.

In case of $\Sigma_1$, the directions correspond directly to the data, and thus,
all pairwise distances in $\cT \cup \cT^\prime$ need to be computed. 
As was shown in \cite[Theorem 3.5]{OwenProvan11}, computing the distance between two trees 
is of complexity $\mathcal{O}(N^4)$. 
Thus, we obtain a complexity of $\mathcal{O}(N^4 \cdot (n + n')^2)$ for $\Sigma_1$.

\end{remark}

Due to the decomposition from Theorem \ref{thm:tancone}, we can at once use Test \ref{test:two-sample-star} to construct the following generally applicable test.

\begin{test}[For equality in the presence of stickiness to strata]\label{test:two-sample-strata}
   Given two independent samples $\cT=\{T_1,\ldots,T_n\}, T_i \iid \prb$ ($i=1,\ldots,n$) and $\cT^\prime=\{T^\prime_1,\ldots,T^\prime_{n'}\}, T^\prime_j \iid \prb^\prime$ ($j=1,\ldots,n$'), and a level $0<\alpha < 1$,
     let $\ort\subset \bhv$ be a stratum of codimension $l\geq 1$, $\mean \in \ort$ and suppose 
     $\mathfrak{T}_T \cong \RR^{N-l-2} \times \mathbb{T}_{k_1} \times \ldots \times \mathbb{T}_{k_m}$ as in Theorem \ref{thm:tancone} with canonical projections $\varpi_r : \mathfrak{T}_T \to \TT_{k_r}, r=1,\ldots,m$.
     Then, we obtain samples $(\cT^r, {\cT'}^r) = (\{T_1^r,\ldots, T_n^r\}, \{{\cT'_1}^r,\ldots, {\cT'_{n'}}^r\})$, $r=1,\ldots,m$, where
     \begin{align*}
        T_i^r &= \varpi_r(\log_\mean(T_i))\,, \qquad i=1,\ldots,n\,,\\ 
        {\cT'_j}^r &= \varpi_r(\log_\mean(T'_j))\,,\qquad j=1,\ldots n'\,,
     \end{align*}
     Letting $\mathrm{p}_r$ be the p-value obtained from conducting Test \ref{test:two-sample-star} for the pair of samples 
     $(\cT^r, {\cT'}^r)$, $r =1,\ldots, m$, reject $\cH_0$ if there exists $r \in \{1,2\ldots,m\}$
     such that
     \[
        \mathrm{p}_{(r)} < \frac{\alpha}{m+1-r}\,,
     \]
    where $\mathrm{p}_{(1)},\ldots,\mathrm{p}_{(m)}$ are the p-values sorted from lowest to highest, see \cite{holm-correction}.
\end{test}


\begin{remark}

    1) Note that $\mu\in \ort$ can be arbitrarily chosen.

    2) In \cite{brown2018mean}, a two sample test was proposed for discriminating
    two probability distributions in BHV spaces based on the difference of their
    Fr\'echet means. If the both distributions are sticky,  however, at the star tree, say,
    the sample Fr\'echet means will be almost
    surely exactly at the star tree beyond some random finite sample size (see Theorem \ref{thm:stickyorth}), rendering their approach infeasible in case.

    3) While Tests \ref{test:two-sample-star} and \ref{test:two-sample-strata} have been motivated by the phenomenon of stickiness, they are, of course, applicable for general distributions.  In particular, they are very meaningful for cases where both sample Fr\'echet means are close to the same stratum (yielding \emph{finite sample stickiness} as discussed by  \cite{ulmer2023exploring-GSI}).
\end{remark}



\section{Applications and Simulations}
\label{sec:examples}
\subsection{Apicomplexa}
\label{sec:api}
\begin{figure}[!h]
	\centering
    \includegraphics[width=.4\linewidth]{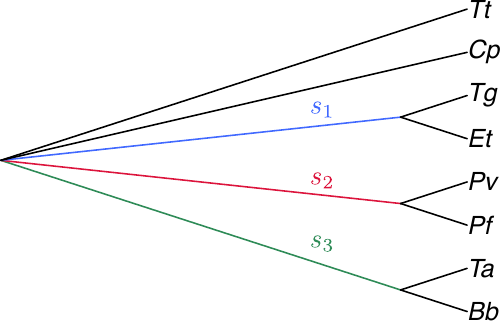}
	\caption{\it The  topology 
 asserted in \cite{bhvpca2}.}
	\caption{\it Proposed topologies of the Fr\'echet mean of
		the apicomplexa data. The abbreviations stem from the species' binary nomenclature.}

	\label{fig:api_top}
\end{figure}

Apicomplexa are a phylum of parasitic alveolates containing a number of important pathogens, such as the causative agents of malaria and toxoplasmosis. The data set we investigate here, originally presented by \cite{api1} and analyzed by \cite{api2},\cite{bhvpca2}, consists of 252 rooted trees with 8 taxa (leaves): \textit{C. parvum} (Cp), 
\textit{T. thermophila} (Tt), \textit{T. gondii} (Tg), \textit{E. tenella} (Et),
\textit{P. vivax} (Pv), \textit{P. falciparum} (Rf),
\textit{B. Bovis} (Bb) and \textit{T. annulata} (Ta).

In particular \cite{bhvpca2} determined the Fr\'echet mean by running Bac{\'a}k's algorithm and
pruning small edges from the output, resulting in a not fully topology, comprising three splits, namely $s_1,s_2,s_3$ from Table \ref{tab:splits}. The topology is also shown in Figure \ref{fig:api_top}. 
{\def\arraystretch{2}\tabcolsep=10pt
\begin{table}[!ht]
	\centering
	\begin{tabular}{ c|c|c} 
        split  & directional derivative & p-value\\
        \hline
		$s_1=\mathrm{Tg},\mathrm{Et} \ \big\vert \  \mathrm{Tt},\mathrm{Cp},\mathrm{Pv},\mathrm{Pf},\mathrm{Ta}, \mathrm{Bb}$ & $\nabla_{\sigma_{s_1}} F_\mathcal{T}(\stree)\approx -1.2\cdot10^{-4}$ & $ < 10^{-15}$\\ 
		$s_2=\mathrm{Pv},\mathrm{Pf} \ \big\vert \  \mathrm{Tt},\mathrm{Cp},\mathrm{Tg},\mathrm{Et},\mathrm{Ta}, \mathrm{Bb}$ & $\nabla_{\sigma_{s_2}} F_\mathcal{T}(\stree) \approx -3.4\cdot10^{-4}$ & $< 10^{-15}$\\ 
		$s_3=\mathrm{Ta},\mathrm{Bb} \ \big\vert \  \mathrm{Tt},\mathrm{Cp},\mathrm{Tg},\mathrm{Et},\mathrm{Pv}, \mathrm{Pf}$ & $\nabla_{\sigma_{s_3}} F_\mathcal{T}(\stree) \approx -2.1\cdot10^{-1}$ & $< 10^{-15}$\\ 
	\end{tabular}
	\caption{\it Splits (first column) comprising the Fr\'echet mean of the apicomplexa data set from \cite{api1} in Section~\ref{sec:api}; splits $s_1,s_2,s_3$ have been found by \cite{bhvpca2}. The second column lists their directional derivatives (negative values imply presence in the mean). 
    The last column shows the    p-values of the corresponding one-sample tests Test \ref{test:one-sample} first three tests were numerically indistinguishable from 0.}\label{tab:splits}
\end{table}
}
We verified the presence of these splits in the sample Fr\'echet mean by showing that the corresponding directional derivatives are negative at the star tree (cf. Theorem~\ref{thm:main-strata}).


Algorithms \ref{alg:prox} and \ref{alg:mindegrees} were run for a hundred different initial orthogonal directions pointing to data, where these directions have been projected according to Theorem \ref{thm:tancone}.
They were not able to find a negative directional derivative, indicating that
the non-fully resolved topology $E(\mu) = \{s_1,s_2,s_3\}$ for the Fr\'echet mean $\mu$ is indeed correct. Their outputs are displayed in Figure \ref{fig:api_alg}.

We used the implementation of Sturm's algorithm of \cite{miller2015polyhedral} 1000 times on our data set with default configurations. 
Despite the Fr\'echet mean being unresolved, all 1000 proposals for the Fr\'echet mean were fully resolved,
which is not coming as a surprise in light of Theorem \ref{thm:sturms-algo-not-singular}.
The splits $s_1,s_2,s_3$ were present in all 1000 outputs. Besides these three splits, we observed 8 other splits in total.
Amongst them, two splits stood out: the splits
$\mathrm{Tt},\mathrm{Cp}\ \big\vert \  \mathrm{Tg},\mathrm{Et}, \mathrm{Pv},\mathrm{Pf},\mathrm{Ta}, \mathrm{Bb}$ and 
$\mathrm{Tg},\mathrm{Et}, \mathrm{Tt},\mathrm{Cp}\ \big\vert \  \mathrm{Pv},\mathrm{Pf},\mathrm{Ta}, \mathrm{Bb}$
were present in 990 and 986 out of 1000 trees, respectively, despite not being in the correct topology we determined.
The other splits were featured less than ten times. This highlights that caution is required when dealing with 
the output of Sturm's algorithm, as even multiple runs on the same data can lead to misleading conclusions.

After having determined the topology of the sample mean, we infer on the topology of the population mean. To this end
we conduct Test \ref{test:two-sample-star} three times at level $0.05$ to test the hypotheses
\begin{align*}
	\mathcal{H}_1^\prime & :   \nabla_{\sigma_{s_1}} F_\prb(\stree) \geq 0\,, \quad
	\mathcal{H}_2^\prime   :   \nabla_{\sigma_{s_2}} F_\prb(\stree) \geq 0\,, \quad
	\mathcal{H}_3^\prime   :   \nabla_{\sigma_{s_3}} F_\prb(\stree) \geq 0\,. 
\end{align*}
After a Holm's correction \cite{holm-correction}, we can reject all three hypotheses at level 0.05.
By Theorem \ref{thm:main-strata}, this implies the presence of splits $s_1, s_2, s_3$ in the population
Fr\'echet mean.


\begin{figure}[!h]
	\centering
	\subfloat[\it Results from running Algorithm \ref{alg:prox}.]{\includegraphics[width=7cm]{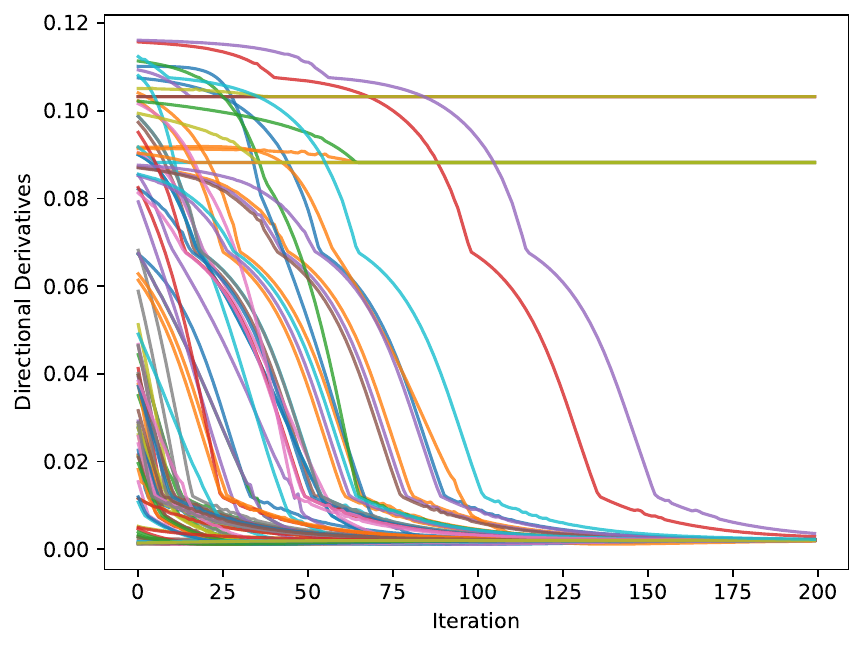}}\qquad
	\subfloat[\it Results from running Algorithm \ref{alg:mindegrees}.]{\includegraphics[width=7cm]{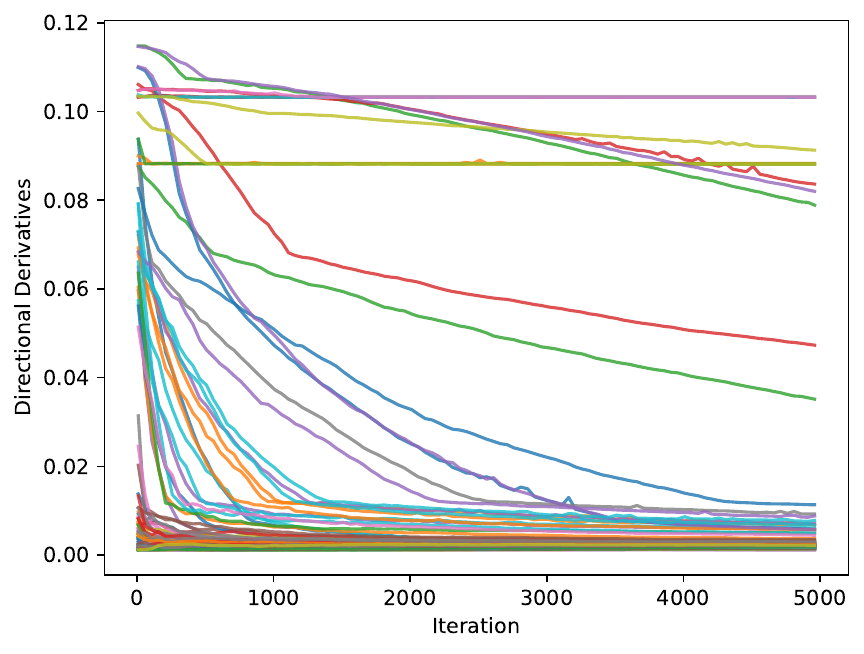}}\qquad

	\caption{\it Directional derivatives for the apicomplexa data set. Each colored line represents starting at one of the first 100 trees: no directions with negative derivative were found.}

	\label{fig:api_alg}
\end{figure}

\subsection{Brain Arteries and Cortical Landmarks}

In \cite{skwerer2014tree}, brain artery trees were analyzed by mapping them to points in
BHV tree space. In human brains, (usually) four brain artery trees emerge from the
circle of Willis. They reconstructed these subtrees from Magnetic Resonance images
and artificially connected these to a root, creating a single brain artery tree.
Furthermore, 128 labelled correspondence points on
the cortex were determined and then connected to the closest vertex in the brain artery tree.
All non-labeled leaves were then pruned, resulting in rooted 85 trees with 128 labelled leaves.
Computing BHV sample Fr\'echet means via Sturm's algorithm,
they observed that the output was very close to the star tree but did not converge, even 
after 50,000 iterations. Rather the topology frequently changed in the last iterations, indicating
the star tree to be the mean of the data.

S. Skwerer and J. S. Marron provided 84 of such trees (one had to be removed) which we split into two data sets 
corresponding to 41 male and 44 female patients. We performed
Algorithm \ref{alg:mindegrees} on both data sets as shown in Figure \ref{fig:brain_data}.
Our analysis suggests that the Fr\'echet mean of both data sets is indeed the star tree. Our two sample tests could not detect a difference between male and female brain trees, which in the light of  \cite{} finding a significant but not highly significant difference by topological data analysis methods suggest that sex differences seem less distinct.

\begin{figure}[!h]
	\centering
	\subfloat[\it Results from running Algorithm \ref{alg:mindegrees} for the female patients.]{\includegraphics[width=7cm]{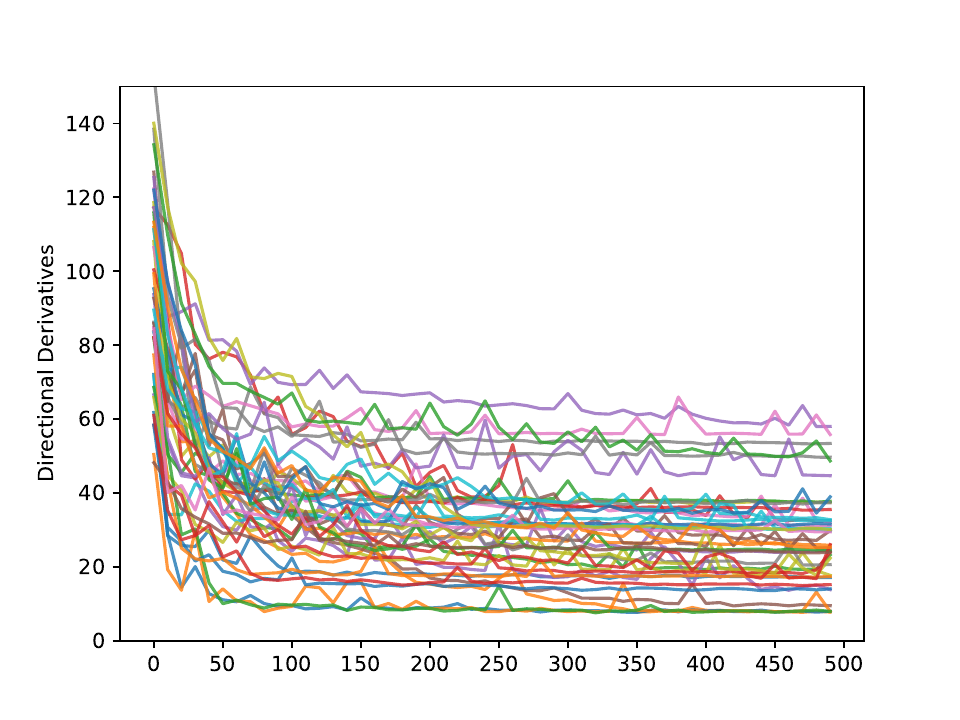}}\qquad
	\subfloat[\it Results from running Algorithm \ref{alg:mindegrees} for the male patients.]{\includegraphics[width=7cm]{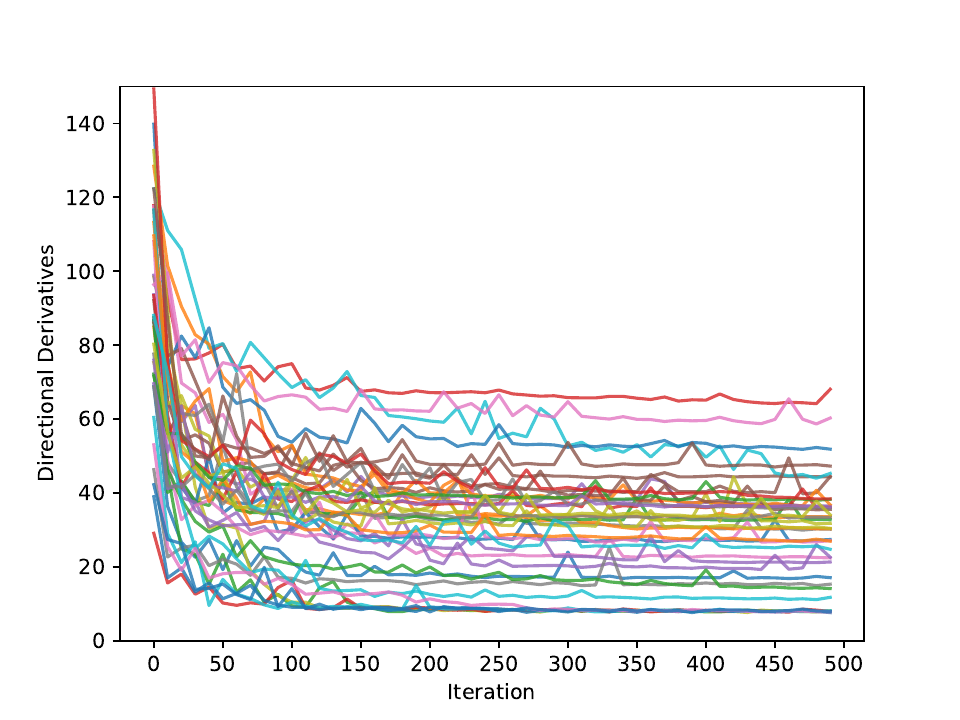}}\qquad

	\caption{\it Due to the rather high leafcount of 128, we chose to perform multiple runs of 
    Algorithm \ref{alg:mindegrees} for different initial positions, c.f. Remark \ref{rem:whichalg}.
    The results support the claim that the star tree is the Fr\'echet mean
	of both data sets corresponding to male and female patients: we did not find negative directional derivatives.}

	\label{fig:brain_data}
\end{figure}

\subsection{Stickiness and Hybridization in Baboon Populations}

We saw in Theorem \ref{thm:stickyorth} that the Fr\'echet mean of a distribution sticks
to a lower-dimensional stratum if and only if the directional derivatives in directions
perpendicular to the stratum are non-negative. This, by Corollary \ref{cor:perpderivs},
is equivalent to stickiness at the star tree in lower-dimesional tree space. 

By Theorems \ref{thm:optimality} and \ref{thm:stickyorth}, a distribution 
$\prb \in \wst[2]{\bhv}$ sticks to the star tree $\stree \in \bhv$ if and only if
\begin{align*}
	0 <  \int - d(\stree, T) \cdot \cos(\angle_\stree(\sigma, T)) \, d\prb(T) \quad \forall \sigma \in \Sigma_\stree \,.
\end{align*}
Therefore, distributions that are spread across multiple orthants of very different topologies, will generally tend to have negative directional derivatives.

A biological factor that could lead to such a distribution is hybridization.
Hybridization (or interbreeding between populations) gives rise to lateral gene transfer 
between different species. A hybridization event
is often modelled as a subtree prune and regraft (SPR) operation on phylogenies: a subtree is removed
from a tree and reattached to another edge (see e.g. \cite{subtreeops} or \cite{stprhyb}).

Depending on the gene sequence and the species involved, the effect of hybridization
on the inferred phylogenetic tree can lead to trees that are very close in
terms of the SPR distance (see e.g. \cite{sprdist}), but far away in BHV distance.
To undo such a SPR operation in the BHV space, one would need to shrink
at least all splits connecting the subtree to its previous position, before regrowing the now missing edges.
An effect of hybridization in a data set of trees can then be that the Fr\'echet mean
becomes non-fully resolved and, most likely, sticky.

The data set we investigate here is a collection of 4260 trees with 19 taxa.
18 of the taxa are baboon populations, and the 19th taxon is an \emph{outgroup}. The
Fr\'echet mean of the data set was found to be the star tree, even failing to detect
the outgroup. The data set was provided by the DPZG and was previously analyzed in \cite{baboons}. 
There, it was found that hybridization occured between the baboon populations, most likely causing
topological heterogeneity between gene trees, and leading to the unresolved Fr\'echet sample mean.
Nevertheless, the directional derivatives of the Fr\'echet function at the star tree
might still be useful in the presence of hybridization.

We computed the median of the overall lengths of the trees and split the data set accordingly into trees from slower- and faster-evolving genetic loci.
The output of Sturm's algorithm for both data sets appeared to be the star tree, motivating a two sample test.
We conducted Test \ref{test:two-sample-star} based on
the directional derivatives of the Fr\'echet function at the star tree. We considered
directions corresponding to single splits and performed 49998 random permutations.
The p-value was estimated to be $\approx 2 \cdot 10^{-5}$, indicating that the
two groups differ significantly, see Figure \ref{fig:baboons}. 
Although this division of the data set by total evolutionary length is artificial, it illustrates how our two-sample test can be applied to experimental data.
\begin{figure}
	\centering
	\includegraphics[scale=.45]{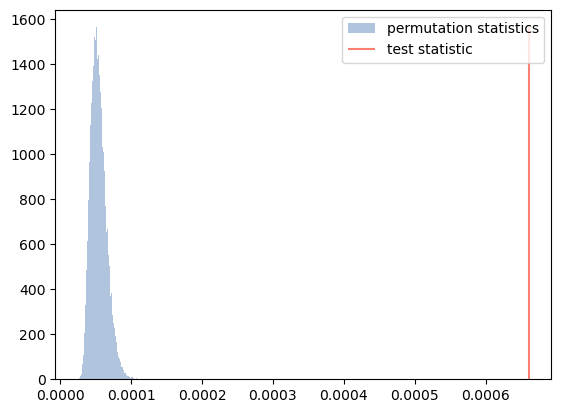}
	\caption{A histogram for the values of the test statistic
		after 49998 random permutations of the 4260 trees.}
	\label{fig:baboons}
\end{figure}

\subsection{Simulations for the Two-Sample Test}

For the two-sample test, we conducted a number of simulations to investigate the power of the test
in different scenarios. We compared the different choices of directions against each other and
against a permutation test based with the Wasserstein distance as test statistic.

Using the Wasserstein distance as test statistic is motivated by the work of
\cite{stickyflavs_arxiv} and 
\cite{bhvsticky}, where the authors proved that stickiness
of distributions corresponds to robustness of the Fr\'echet mean against small changes
in the Wasserstein space of probability distributions, i.e. 
every other distribution that is sufficiently close in Wasserstein distance has a Fr\'echet mean lying on the same stratum.

We can observe in the experiments displayed in Figure \ref{fig:power1}, that tests
based on the directional derivatives appear to be more powerful than tests based on the Wasserstein distance.
Figure \ref{fig:power1a} shows that this result depends considerably on the choice of directions,
all choices outperform the Wasserstein distance in the more anisotropic case of Figure \ref{fig:power1b}.

\begin{figure}[h]%
	\centering
	\subfloat[\it
		The x-axis displays the
		scales chosen for drawing edge lengths the second sample.]{\includegraphics[scale=.45]{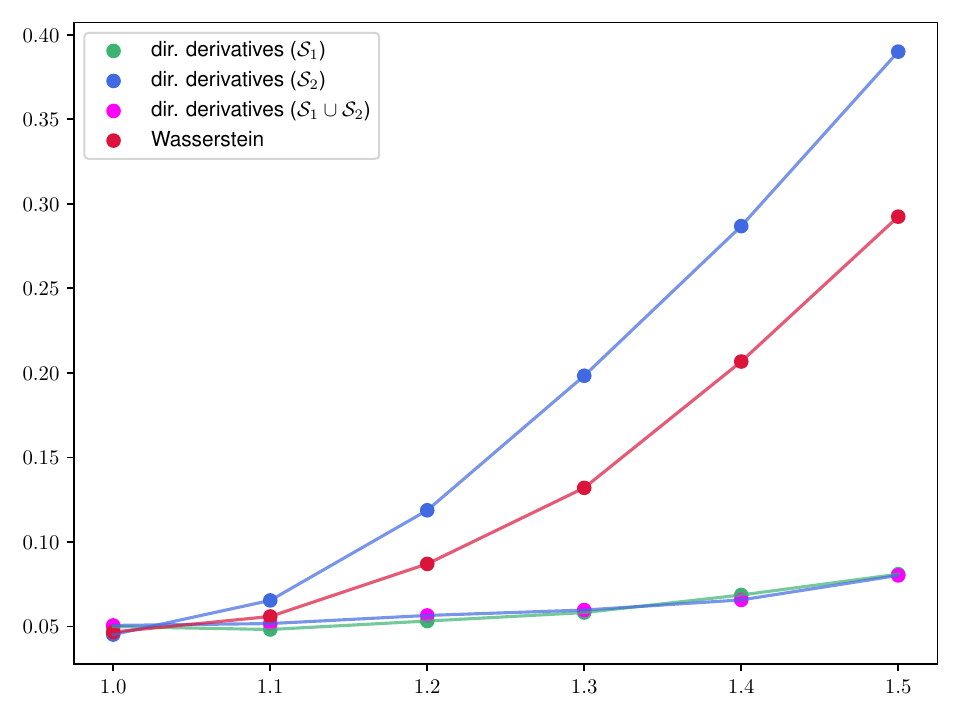}\label{fig:power1a}} \hspace{.15cm}
	\subfloat[\it
		For eight predetermined topologies, edges are drawn from exponential distributions with scale 1.
		The chosen scale for the remaining seven topologies is displayed on the x-axis.]{\includegraphics[scale=.45]{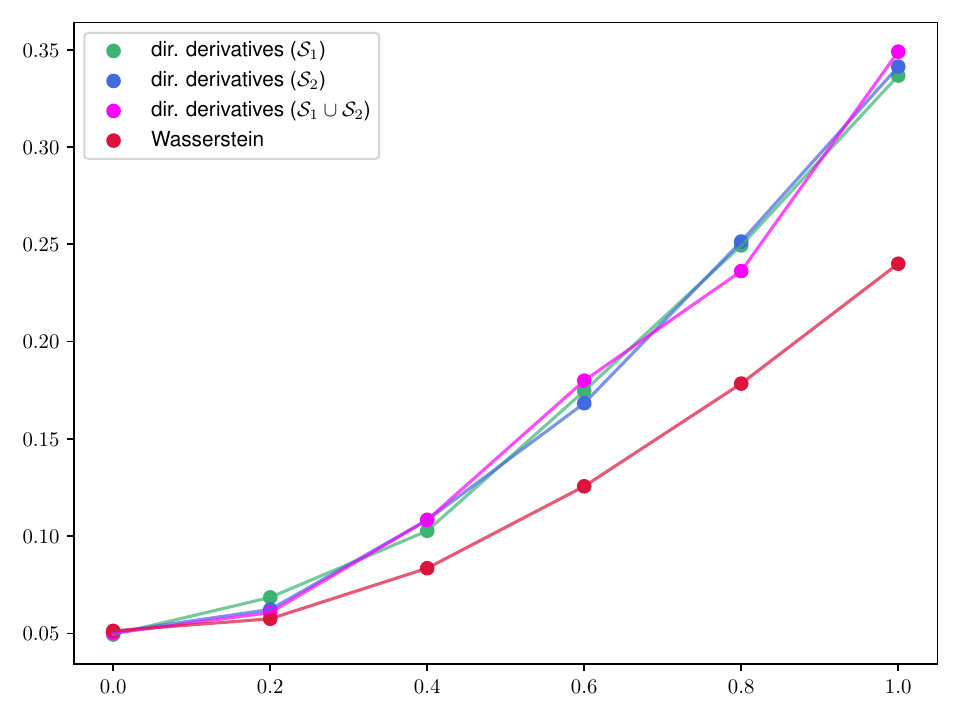}\label{fig:power1b}}
	\caption{\it In both numerical experiments, two samples in $\cT, \cT' \subset \mathbb{T}_4$
		were generated, both having an equal sample size $\lvert \cT \rvert =
			\lvert \cT' \rvert = 50$. For both samples, the topologies were drawn uniformly and
		the edge lengths in the first sample follow exponential distributions with scale 1.
		Figure \ref{fig:power1a} and Figure \ref{fig:power1b} correspond to two different
		ways of drawing the edge lengths of the second sample. The level for all tests
		was set to $0.05$.}
	\label{fig:power1}%
\end{figure}

It appears, however, that the test is not optimal for discrete probability
distributions as can be seen in Figure \ref{fig:power2}.
In this scenario, the test based on the Wasserstein distance easily outperforms
the tests based on directional derivatives.
For such distributions, a Wasserstein based test might be more appropriate.

We also want to highlight that our test has a computational advantage.
For each permutation, the evaluation of the Wasserstein distance requires
solving a linear program, whereas our test statistic only requires summation and
subtraction of the permuted directional derivatives.

\begin{figure}
	\centering
	\includegraphics[scale=.5]{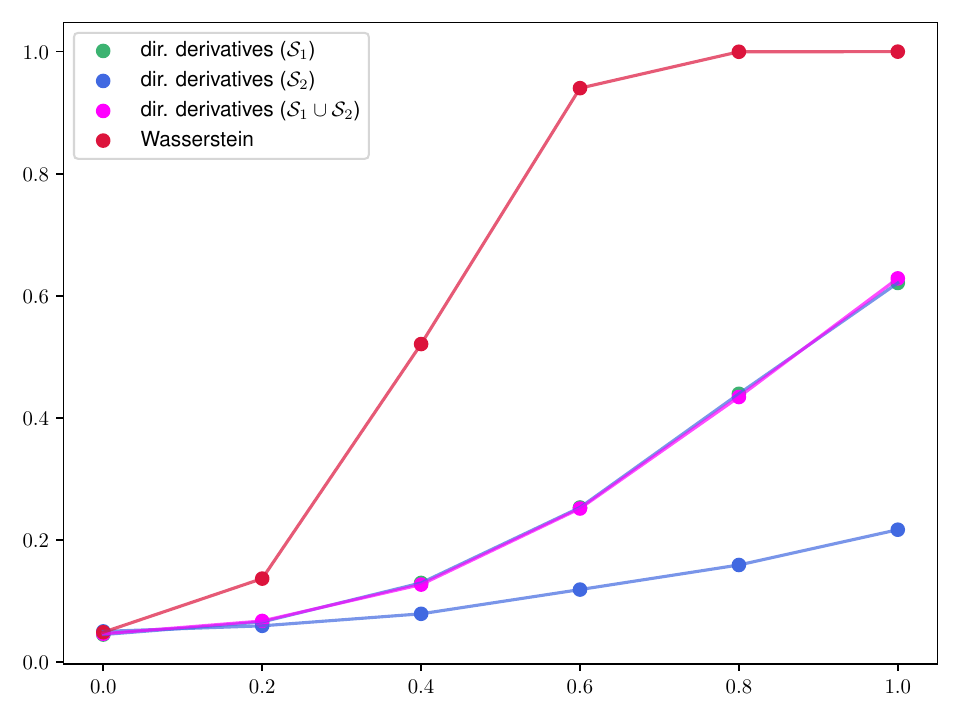}
	\caption{\it In these numerical experiments, two samples in $\cT, \cT' \subset \mathbb{T}_4$
		were generated, both having an equal sample size $\lvert \cT\rvert =
			\lvert \cT' \rvert = 50$. For both samples, the edge lengths are fixed at
		1. The topology of the first sample were drawn uniformly.
		For the second sample, three predetermined topologies were drawn with
		probability $(1 + \lambda \cdot 4)/15$, the remaining topologies with
		probability $(1 - \lambda) /15$. The parameter was chosen from $\lambda \in \{0, 0.2,0.4,0.6,0.8,1\}$.}
	\label{fig:power2}
\end{figure}

\section{Discussion}
\label{sec:discussion}
%

In this paper we have illustrated that
despite the attractive features of BHV tree space, use of the Fr\'echet mean can be challenging due to (i) computational issues calculating the sample mean and (ii) stickiness of the population mean affecting asymptotic tests. 

 The new sufficient condition we have derived for presence of edges in the Fr\'echet mean tree applies to both sample and population means. 
As illustrated in Example \ref{ex:improv}, this condition is a strict improvement on the condition in \cite[Theorem 1]{bhvmean}.

 Our rigorous test and algorithm for identifying splits in  the population and sample mean is a statistically valid method as opposed to arbitrary removal of short edges, illustrated by the Apicomplexa example and the brain data example. 

 Further, our new two-sample test can distinguish distributions which share the same sticky mean, illustrated by the baboon example. 


 The baboon data also suggests a link between lack of resolution and underlying biological mechanisms such as 
populations believed to have undergone extensive hybridization. 
The effect on gene trees can be represented by SPR which results in highly dispersed samples in BHV tree space and unresolved Fr\'echet means. 
Although lack of resolution in the Fr\'echet mean suggests a loss of information (e.g. the sample mean does not move when data change), the directional derivatives of the Fr\'echet function may retain information about the biological processes relating species trees to gene trees. 


\bibliographystyle{Chicago}
\bibliography{references}

\newpage
\appendix

\section{Proofs}
\label{sec:proofs}

\subsection{Proof of Theorem \ref{thm:sturms-algo-not-singular}}
\label{sec:proof_sturm}

\begin{proof}

	Let $Y_j, j \in \NN$ be the sequence of i.i.d. random variables that is used to generate the inductive means $\hat \mu_j$ of Sturm's algorithm. 
   By hypothesis there is a split $s$ such that
    $$ \cT_s : =\{T\in \cT: E(T) \ni s \in C(\mu) \setminus E(\mu)\}\neq \emptyset\,.$$
   Further, by construction, if $\hat{\mean}_{k-1} \in \ort$ and $Y_j \in \mathcal{T}_s$. i.e. $\vert s\vert_{Y_j} = \lambda$ for some $\lambda>0$, then $\vert s\vert_{\hat{\mean}_j} = \frac{\lambda}{j+1}$ by \eqref{eq:bhv_geod_splits}, 
   i.e. $\hat{\mean}_j \notin \ort$. In consequence
   \begin{eqnarray}\label{eq:sturms-algo-oscialltes-proof} \{\hat\mu_{j-1}\in \ort, Y_j \in \cT_s \}&\subseteq& \{\hat{\mean}_j \notin \ort\} \,.\end{eqnarray}
    Since 
	\begin{align*}
		\sum_{j=1}^{\infty} \mathbb{P}\{Y_j \in \mathcal{T}_s\} = \sum_{j=1}^{\infty} \frac{\lvert \mathcal{T}_s\rvert}{\vert \cT\vert} = \infty\,,
	\end{align*}      
      application of the Borell-Cantelli Lemma, e.g.  \citet[Theorem 2.3.7]{durrett}, yields in conjunction with (\ref{eq:sturms-algo-oscialltes-proof}), as asserted:
      $$ 1 = \mathbb P\{Y_j \in \cT_s\, \infty\mbox{ often}\} \leq \mathbb P \{\hat\mu_{j}\not\in \ort\, \infty\mbox{ often} \}\,. $$
 

\end{proof}

\subsection{Proof of Theorem \ref{thm:main-strata}}
\begin{figure}
	\centering
	\subfloat[\it Case 1.]{\includegraphics[width=7cm]{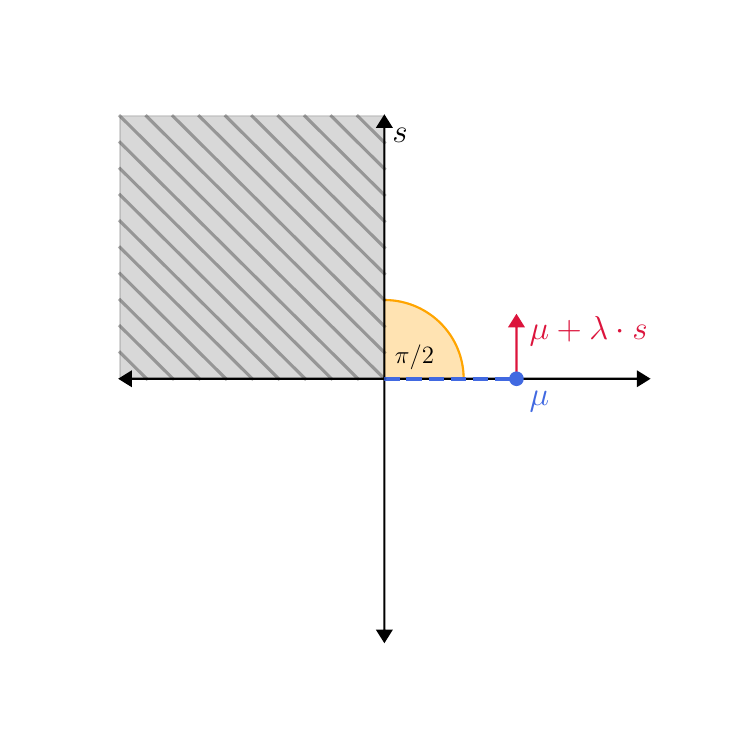}}\qquad
	\subfloat[\it Case 3.]{\includegraphics[width=7cm]{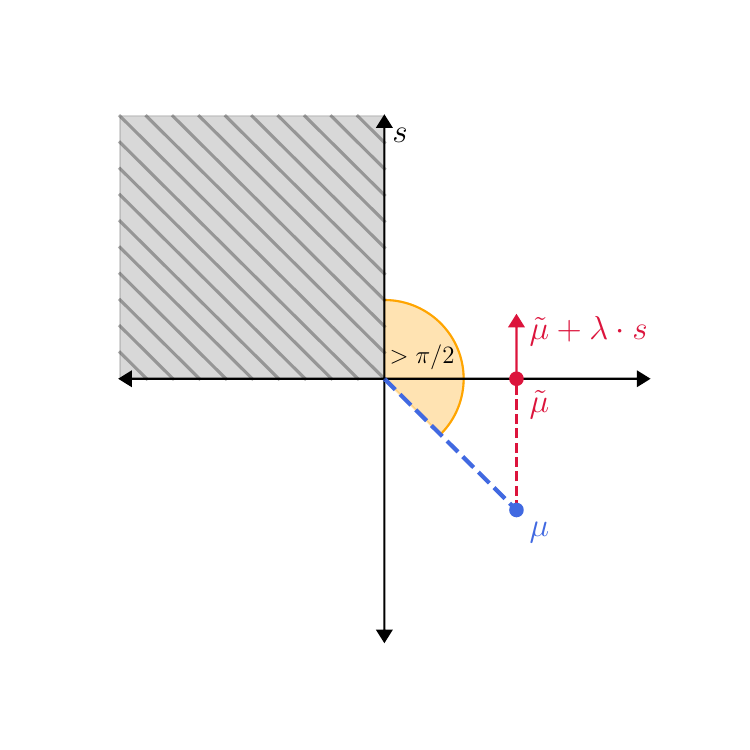}}\qquad

	\caption{\it An illustration of cases 1 and 3 in the proof of Theorem \ref{thm:main-strata} in $\mathbb{T}_4$.}
	\label{fig:proof}
\end{figure}


    We first show that (i) implies (ii). Indeed, negating the statement in (i), we obtain $\nabla_{\sigma_s} F_\prb (T) < 0 \implies \angle_T(\sigma_s, \mean) < \pi/2$.
    This, in turn, implies that $\mean$ lies in an stratum bordering $T + 1 s$. Since $E(T) \subset E(\mean)$, trees in the 
    stratum of $\mean$ must contain $s \in E(\mean)$, yielding (ii).

    Next we show (i).
	By 
    Lemma~\ref{lemma:const_deriv}, the directional derivative 
	$\nabla_{\sigma_s} F_\prb(T)$ does not depend on the position of 
	$T \in \ort$. Therefore, we can assume that 
	$T = \prj_{\overline{\ort}} (\mean) \in \ort$.
    We distinguish between three cases, see Figure~\ref{fig:proof} for visualisations
    of cases 1 and 3.

	\noindent \textbf{Case 1:} $\angle_T(\mean, s) = \pi/2$.

	If $\angle_T(\mean, \sigma_s) = \pi/2$, the topology of $\mean$ cannot be
	fully resolved and $s \in C(\mean)$. Therefore, it follows that
	$\nabla_{\sigma_s} F_\prb(\mean) \geq 0$ by 
    Theorem~\ref{thm:optimality}.
	Now, let $\gamma_T^\mean$ denote the unit speed
	geodesic from $T$ to $\mean$.
	By Lemma~\ref{lemma:const_deriv},
    we also have for every $\kappa \in (0, d(T, \mean)]$
	that $\nabla_{\sigma_s} F_\prb(\gamma_T^\mean(\kappa)) = 
	\nabla_{\sigma_s} F_\prb(\mean) \geq 0$. 
	Since  $F_\prb$ is convex, e.g. \cite[proof of Proposition 4.3]{sturm}, the derivative traversing 
	from $\gamma_T^\mean(\kappa)$ in direction $\sigma_s$ is nondecreasing, 
	i.e. for all $\lambda >0 $ and all $\kappa \in (0, d(T, \mean)]$, we 
	have $F_\prb(\gamma_T^\mean(\kappa)) \leq F_\prb(\gamma_T^\mean(\kappa) + \lambda \cdot s)\,.$
	Taking the limit $\kappa \to 0$, we obtain by continuity of $F_\prb$ that 
	$F_\prb(T) \leq F_\prb(T + \lambda \cdot s)$.
	Thus,
	\begin{align*}
		\nabla_{\sigma_s} F_\prb(T) = \lim_{\lambda \searrow 0} \frac{F_\prb(T + \lambda \cdot s) - F_\prb(T)}{\lambda} \geq 0 \,.
	\end{align*}

	\noindent \textbf{Case 2:} $\angle_T(\mean, \sigma_s) = \pi$.

	In this case, the unit speed geodesic $\gamma_\mean^{T + 1 \cdot s}$ from $\mean$
	to $T + 1 \cdot s$ passes through the star tree. As $\mean$ is the
	unique minimizer of the convek $F_\prb$, the function $F_\prb \circ
		\gamma_\mean^{T + 1 \cdot s}$ is non-decreasing. Consequently,
	we have that
	\begin{align*}
		\nabla_{\sigma_s} F_\prb(T) = \lim_{\lambda \searrow 0} \frac{F_\prb\big(\gamma_\mean^{T + 1 \cdot s}(d(T, \mean) + \lambda)\big) - F_\prb(\overbrace{\gamma_\mean^{T + 1 \cdot s}(d(T, \mean)}^{=T})}{\lambda} \geq 0 \,.
	\end{align*}

	\noindent \textbf{Case 3:} $\pi / 2 < \angle_T(\mean, \sigma_s) < \pi$.

	In this case, $s$ must be compatible with some the splits in $E(\mean) 
	\setminus E(T)$ but not with all of them. Consider the tree $\tilde{\mean}$ 
	obtained by removing the splits of $\mean$ incompatible with $s$. 
	At $\tilde{\mean}$, we have that $\angle_{\tilde{\mean}}(\sigma_s, \mean) = 
	\pi$ due to the incompatibility of $s$ and the deleted splits. 
	Since $\angle_{\tilde{\mean}}(\sigma_s, \mean) = \pi$, the geodesic
	$\gamma_{\mean}^{\tilde{\mean}+1\cdot s}$ must pass through $\tilde{\mean}$.
	Furthermore, since $\mean$ is the unique Fr\'echet mean, 
	$F_\prb \circ \gamma_{\mean}^{\tilde{\mean}+1\cdot s}$ is non-decreasing. 
	Hence, $\nabla_s F_\prb(\tilde{\mean}) \geq 0$.
 
	Next, consider the geodesic $\gamma_T^{\tilde{\mean}}$. Again, we
	have that for every $\kappa \in (0, d(T, \tilde{\mean})]$ that
	$\nabla_{\sigma_s} F_\prb(\gamma_T^{\tilde{\mean}}(\kappa)) = \nabla_{\sigma_s} F_\prb(\tilde{\mean}) \geq 0$.
	Consequently, we have for every $\lambda >0$ and $\kappa \in (0, d(T, \tilde{\mean})]$
	by the convexity of $F_\prb$ that $F_\prb(\gamma_T^{\tilde{\mean}}(\kappa)) \leq F_\prb(T + \lambda \cdot s)$
	Taking the limit $\kappa \to 0$, we see $F_\prb(T) \leq F_\prb(T + \lambda \cdot s)$
	for every $\lambda >0$. Thus, $\nabla_{\sigma_s} F_\prb(T) \geq 0$, which is assertion (i).
\qed

\subsection{Proof of Corollary \ref{cor:sample}}

	The Fr\'echet function at the star tree is given by
	\begin{align*}
		\frf (\stree)= \frac{1}{2 \cdot \lvert \mathcal{T} \rvert}
		\sum_{T \in \mathcal{T}} \sum_{x \in E(T)}
		\lvert x \rvert_T^2\,.
	\end{align*}
	Furthermore, by Remark~\ref{rem:dist_single_split}
	\begin{align*}
		2 \cdot \lvert \mathcal{T} \rvert \cdot \frf (\stree + \lambda \cdot s) = \quad &
		\sum_{\footnotesize\begin{array}{c}T \in \mathcal{T}\\ s \in E(T)\end{array}} \left( (\lambda - \lvert s \rvert_T)^2
		+ \sum_{s\neq x \in E(T)} \lvert x \rvert_T^2 \right)
		+ \sum_{\footnotesize\begin{array}{c}T \in \mathcal{T}\\ s \in C(T) \\
  s\not\in E(T)\end{array}} \left( \lambda^2 +
		\sum_{x \in E(T)} \lvert x \rvert_T^2 \right)                                                 \\
		+                                                                               & \sum_{\footnotesize\begin{array}{c}T \in \mathcal{T}\\ s \not\in C(T) \end{array}} \left(
		\left(\lambda + \sqrt{\sum_{C(s)\not\ni x \in E(T)
  }
				\lvert x \rvert_T^2}\right)^2 + \sum_{\footnotesize\begin{array}{c}x \in E(T)\\ x\in C(s) \end{array}}
		\lvert x \rvert_T^2 \right) \,.                                                                                                 \\
	\end{align*}
	Thus, we get
	\begin{align*}
		\frf(\stree + \lambda \cdot s)- \frf(\stree) = \frac{\lambda^2}{2} + \frac{\lambda}
		{\lvert \mathcal{T} \rvert} \cdot \left(\sum_{T \in \mathcal{T}}
		\sqrt{\sum_{C(s)\not\ni x \in E(T)}  \lvert x \rvert_T^2} -
		\sum_{T \in \mathcal{T}} \lvert s \rvert_T \right)\,.
	\end{align*}
	And hence
	\begin{align*}
		\nabla_{\sigma_s} F_\mathcal{T} (\stree) = \frac{1}{\lvert \mathcal{T} \rvert} \left(\sum_{T \in \mathcal{T}}
		\sqrt{\sum_{C(s)\not\ni x \in E(T)} \lvert x \rvert_T^2}
		-\sum_{T \in \mathcal{T}} \lvert s \rvert_T \right)\,.
	\end{align*}
	Theorem \ref{thm:main-strata} then yields the claim.
\qed

\subsection{Proof of Theorem \ref{thm:tancone}}

We first state and show two auxiliary lemmata and then prove the theorem.

\begin{lemma}
	\label{lemma:proj}

	Let $T \in \ort$. If $\bhv \ni T^\prime \in \prj_{\overline{\ort}}^{-1}(\{T\})$, we
	have $\lvert s \rvert_{T} = \lvert s \rvert_{T^\prime}$ for every $s \in E(T)$.
\end{lemma}

\begin{proof}
	Let $(A_0, \ldots A_k)$, $(B_0, \ldots B_k)$ be the support pair of the
	geodesic from $T$ to $T^\prime$.

	If $E(T) \subseteq E(T^\prime)$, due to (\ref{eq:bhv-distance-explicit}), the distance is given by
	\begin{align*}
		d(T, T^\prime) = \sqrt{\sum_{s \in E(T)}  (\lvert s \rvert_{T} - \lvert s \rvert_{T^\prime})^2 + \sum_{s \in E(T^\prime) \setminus E(\tilde{T})}\lvert s \rvert_{T^\prime}^2}\,.
	\end{align*}
	The first term $\sum_{s \in E(T)}  (\lvert s \rvert_{T} - \lvert s \rvert_{T^\prime})^2$
	is minimal if and only if it is 0. Hence, the claim follows.

	Next, assume that $E(T) \not \subseteq E(T^\prime)$. As $E(T) \neq E(T)
		\cap E(T^\prime)$, there is $i \in \{1,2,\ldots,k\}$ such that $\lVert A_i
		\rVert_T > 0$.
	Define $\widetilde{T} \in \overline{\ort}$ with
	$E(\widetilde{T}) = E(T) \cap E(T^\prime)$ by $\lvert s \rvert_{\tilde{T}}
		= \lvert s \rvert_{T^\prime}$ for all $s\in E(\widetilde{T})$.  By definition, it holds $\bigcup_{i=1}^k B_i = E(T^\prime) \setminus E(T) = E(T^\prime) \setminus E(\tilde{T})$.
	Thus, we obtain, again by (\ref{eq:bhv-distance-explicit}),
	\begin{align*}
		d(T, T^\prime) & = \sqrt{\sum_{s \in A_0}  (\lvert s \rvert_{T_1} - \lvert s \rvert_{T_1})^2 + \sum_{i =1}^k(\lVert A_i \rVert + \lVert B_i \rVert)^2}
		\geq \sqrt{\sum_{i =1}^k(\lVert A_i \rVert + \lVert B_i \rVert)^2}                                                                                     \\
		               & > \sqrt{\sum_{i =1}^k \lVert B_i \rVert^2}
		= \sqrt{\sum_{s \in E(T^\prime) \setminus E(\tilde{T})}\lvert s \rvert_{T^\prime}^2}
		= d(\tilde{T}, T^\prime)\,,
	\end{align*}
	a contradiction to the assumption that $P_{\overline{\ort}}(T^\prime) = T \neq \widetilde{T}$.

\end{proof}

We will also need the following Lemma, c.f. \cite[Theorem 2.1]{Owen2011only}.

\begin{lemma}
    \label{lemma:shared_edges}
    
	Let $N\geq4$ and $T_1, T_2 \in \bhv$. Let $\cE \subseteq E(T_1) \cap E(T_2)$.
	  Then, there are lower dimensional BHV spaces $\TT_{k_j}$ of dimensions $k_j$ and pairs of
	trees $(T_1^j, T_2^j)$, $j=1,\ldots,\vert \cE\vert +1$ such that
	\begin{enumerate}[label=(\roman*)]
		\item $T_1^j, T_2^j \in \mathbb{T}_{k_j}$ for all $j \in \{1,2,\ldots,\vert \cE\vert +1\}$,
		\item 
  $\sum_{j=1}^{\vert \cE\vert+1} (k_j - 2) = (N-2)-\vert \cE\vert$,
		\item $d^2(T_1, T_2) = \sum_{s \in \cE}(\lvert s \rvert_{T_1} - \lvert s \rvert_{T_2})^2 + \sum_{j=1}^{\vert \cE\vert +1} d^2(T_1^j, T_2^j)$.
	\end{enumerate}
\end{lemma}

\begin{proof}

    If $\cE = \emptyset$ the assertion is trivial. 
    Hence we now assume that $\lvert\cE\rvert =: m\geq 1$. We will then iteratively construct the tuples $(T_1^j, T_2^j)$, $j=1,\ldots,\cE$ by bisection at common splits $s\in \cE$.
    
    Suppose the root's label is 0. We begin by choosing a split
        \begin{equation}
        \label{eq:bisec_0}
    A\vert B = s \in \argmin_{\substack{C \vert D \in \cE\,,\\ 0 \in D}} \lvert C \rvert \,.
    \end{equation}

    By \cite[Theorem 2.1]{Owen2011only}, bisection of $T_1, T_2$ at  the shared split $s$ yields four trees $(T_1^1, T_2^1),(\widetilde{T}_1^1, \widetilde{T}_2^1)$, where a new leaf $v$ is added in lieu of the deleted split: the set of leaves of $T_1^1, T_2^1$ is given by $A$ and $v$ serves as root, the set of leaves of $\widetilde{T}_1^1,\widetilde{T}_2^1$ is $(B \setminus\{0\})\cup \{v\}$ where $0$ remains the root. Setting $k_1 := \lvert A \rvert$, and $\tilde{k}_1 := \lvert B \rvert - 1 + 1$, \cite[Theorem 2.1]{Owen2011only} showed that      
    $T_1^1,T_2^1 \in \mathbb{T}_{k_1}$, $\widetilde{T}_1^1,\widetilde{T}_2^1 \in  \mathbb{T}_{\tilde{k}_1}$, 
    \begin{equation}
        \label{eq:bisec_1}
        d^2(T_1, T_2) = \left(\lvert s \rvert_{T_1} - \lvert s \rvert_{T_2}\right)^2 + d^2\left(T_1^1, T_2^1\right) + d^2\left(\widetilde{T}_1^1, \widetilde{T}_2^1\right)\,,
    \end{equation}
     and we add at once that
    \begin{equation}
        \label{eq:bisec_2}
        (k_1-2)+(\tilde{k}_1-2) = \lvert A \rvert + \lvert B \rvert - 4 = (N-2) - 1\,.
    \end{equation}
    If $m=1$, setting $k_2 = \tilde{k}_1$, and $T_i^2:= \widetilde{T}_i^1 $ for $i=1,2$,  the assertion follows from (\ref{eq:bisec_1}) and (\ref{eq:bisec_2}).
    
    In case of $m>1$, we choose the next split
        \[A'\vert B' = s' \in \argmin_{\substack{C \vert D \in \cE \setminus\{s\}\,,\\ 0 \in D}} \lvert C \rvert \,.
    \]
    Since $s,s^\prime$ are compatible, $0 \in B\cap B'$ and $ B \cap A^\prime x= \emptyset \overset{A^\prime \subset A \cup B}{\implies} A^\prime \subset A$, yielding a contradiction to minimality of $A$ in (\ref{eq:bisec_0}), only one of the following two cases holds
    \begin{enumerate}
        \item[(a)] $A \cap A^\prime = \emptyset \overset{A \subset A^\prime \cup B^\prime}{\implies} A \subset B^\prime \setminus\{0\}$ and $A^\prime \subset B \setminus\{0\}$
        \item[(b)] $A \cap B^\prime = \emptyset \overset{B \subset A^\prime \cup B^\prime}{\implies} A \subset A^\prime$ and $B^\prime \subset B$.
    \end{enumerate}
    Recalling that $\widetilde{T}_i^1$ are trees over the leaf set $(B\setminus\{0\})\cup\{v\}$ with root $0$, in case of (a) the split $s'$ corresponds to the split $\widetilde{s}^\prime :=A^\prime|(B^\prime \setminus A)\cup\{v\}$ of $\widetilde{T}_i^1$, $i=1,2$; whereas in case (b) it corresponds to $\widetilde{s}^\prime :=(A^\prime \setminus A)\cup\{v\}|B^\prime$. Deleting this split results in a new vertex $v'$ and, once again invoking \cite[Theorem 2.1]{Owen2011only}, we obtain four trees (see Figure \ref{fig:proof-bisection-lem}:

        \begin{figure}
        \centering
        \subfloat[\it]{\includegraphics[width=.45\textwidth]{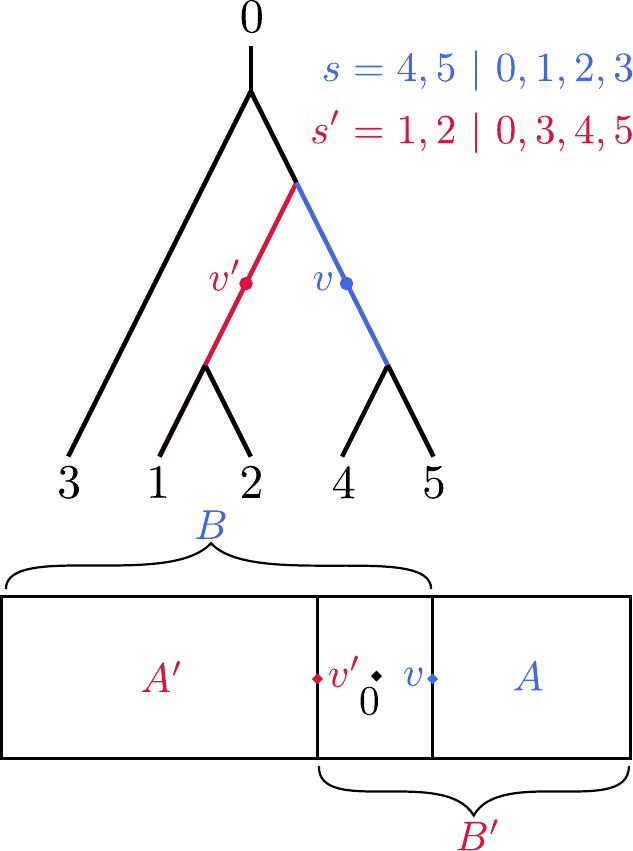}}\qquad
	   \subfloat[\it ]{\includegraphics[width=.45\textwidth]{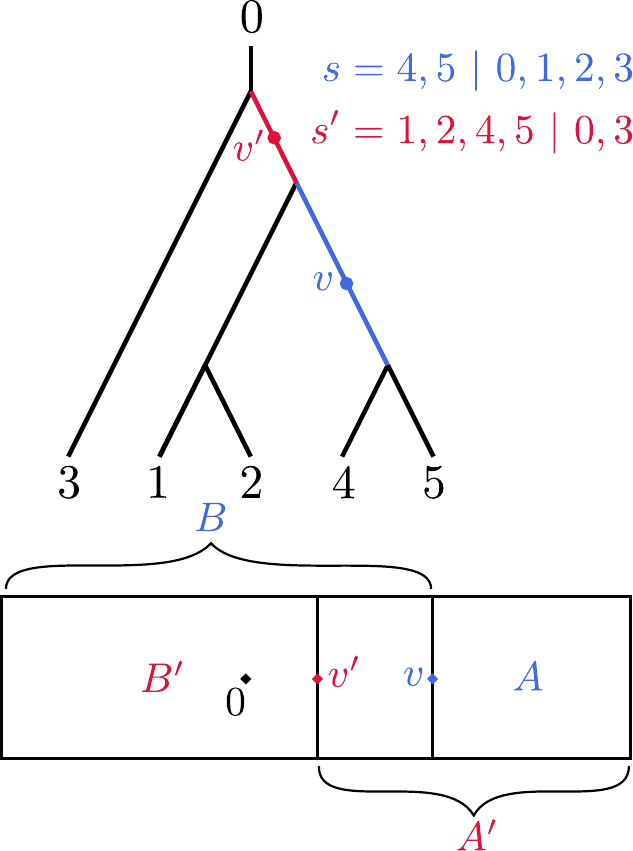}}\qquad
        \caption{\it Schematic display of split sets for the trees in the proof of Lemma \ref{lemma:shared_edges} with examples on the top. Left: case (a), right: case (b)}
        \label{fig:proof-bisection-lem}
    \end{figure}

    The first two are $T^2_1,T^2_2 \in \TT_{k_2}$ with leaf set $A^\prime$ (case (a)) and $(A^\prime \setminus A)\cup\{v\}$ (case (b)), respectively, with new root $v'$, where we have set $k_2 := \lvert A^\prime\lvert$ (case (a)) and $k_2 := \lvert A^\prime\rvert - \lvert A\rvert + 1$ (case (b)), respectively. 
    
    The other two trees are $\widetilde{T}^2_1,\widetilde{T}^2_2\in \TT_{\widetilde{k}_2} $ with leaf set $(B^\prime \cup\{v,v'\})\setminus A $ (case (a)) and $B^\prime \cup\{v'\}$ (case (b)), respectively, with root $0$, where we have set $\widetilde{k}_2 := \lvert B^\prime \rvert - \lvert A \rvert +1$ (case (a)) and $:= \lvert B^\prime\lvert$ (case (b)), respectively. Due to \cite[Theorem 2.1]{Owen2011only}, we obtain

    \begin{equation}
        \label{eq:bisec_3}
        d^2(\widetilde{T}^1_1, \widetilde{T}^1_2) = \left(\lvert \widetilde{s}^\prime  \rvert_{\widetilde{T}^1_1} - \lvert \widetilde{s}^\prime \rvert_{\widetilde{T}^1_2}\right)^2 + d^2\left(T_1^2, T_2^2\right) + d^2\left(\widetilde{T}_1^2, \widetilde{T}_2^2\right)\,,
    \end{equation}
     and at once
    \begin{equation}
        \label{eq:bisec_4}
        (k_2-2)+(\tilde{k}_2-2) = \lvert A' \rvert + \lvert B' \rvert - \lvert A \rvert- 3 = N+1 - k_1 - 3 = N-2 -2 -(k_1-2) \,.
    \end{equation}

    If $m=2$, setting $k_3 := \tilde{k}_2$, and $T_i^3:= \widetilde{T}_i^2 $ for $i=1,2$,  the assertion follows from (\ref{eq:bisec_1}),  (\ref{eq:bisec_2}), (\ref{eq:bisec_3}) and (\ref{eq:bisec_2}), since $\lvert \widetilde{s}^\prime \rvert_{\widetilde{T}^1_i} = \lvert s^\prime \rvert_{T_i}$ for $i=1,2$.

    If $m>2$ iteration of the above step for all remaining splits in $\cE\setminus\{s,s'\}$ yields the assertion.

    \end{proof}

\noindent
{\it Proof of Theorem \ref{thm:tancone}.}
	As shown in \cite[Lemma 2]{bhvsticky}, we have $\Sigma_T = 
    \Sigma_T^\parallel \ast \perpdir$.
	Since $\ort$ is isometric to a Euclidean orthant of dimension $m=N-l-2$, one 	has by the definition of $\Sigma_T^\parallel$ that 	$\Sigma_T^\parallel \cong S^{N-l-3}$. 
	Therefore, we are left to understand the structure of $\perpdir$.
	Let $\sigma_1, \sigma_2 \in \perpdir$ with corresponding 
	geodesics $\overline{\gamma}_1, \overline{\gamma}_2: [0,1]\to\bhv$ 
    starting at $T$ with directions $\sigma_1,\sigma_2$, respectively. 
    We set $T_i = \gamma_i(1)$, $i=1,2$. By definition of $\perpdir$, we have
    that $T_1, T_2 \in \prj_{\overline{\ort}}^{-1}(\{T\})$. We know by 
    Lemma \ref{lemma:proj} that 
    \begin{equation}
    \label{eq:same_lengths}
        \lvert T_1\rvert_s = \lvert T_2 \rvert_s =  \lvert T \rvert_s
    \end{equation}
    for all $s\in E(T)$. In particular, we have $E(T) \subset E(T_1)\cap E(T_2)$.

    Now  Next, we apply Lemma \ref{lemma:shared_edges} (for notation purposes further below, we take recourse to additional tildes here) for $T_1, T_2$ and $\cE = E(T)$, giving us
    pairs of trees $(\widetilde{T}_1^j, \widetilde{T}_2^j) \in \TT_{\tilde{k}_j}$, $j=1,\ldots,\lvert E(T) \rvert +1$ such that
    \begin{align}
        \label{eq:dimension_sum}
        \sum_{j=1}^{\lvert E(T) \rvert + 1} (\tilde{k}_j - 2) = N-2-\lvert E(T) \rvert =l\,,
    \end{align}
    and
    \begin{equation}
        \label{eq:distance_orth}
	       d^2(T_1, T_2) = \sum_{s \in \cE} (\lvert s \rvert_{T_1} - \lvert s \rvert_{T_2})^2 + \sum_{j=1}^{\lvert E(T) \rvert + 1} d^2(\widetilde{T}_1^j, \widetilde{T}_2^j) \stackrel{\eqref{eq:same_lengths}}{=} \sum_{j=1}^{\lvert E(T) \rvert + 1} d^2(\widetilde{T}_1^j, \widetilde{T}_2^j) \,. 
    \end{equation}
    At this point, we might have some $\tilde{k}_j=2$ for some $j \in \{1,\ldots, \lvert E(T) \rvert +1\}$. 
    As $\TT_2$ only contains its star tree one then has $d(\widetilde{T}_1^j,\widetilde{T}_2^j)=0$.
    Thus, neither does such $\tilde{k}_j$ contribute to \eqref{eq:dimension_sum} nor do the trees $\widetilde{T}_1^j,\widetilde{T}_2^j \in \TT_{\tilde{k}_j}$ contribute to \eqref{eq:distance_orth}.
    
    Therefore, we can safely remove every $\tilde{k}_j=2$, yielding $k_1, \ldots k_m$ and $(T_1^j, T_2^j) \in \TT_{k_j}$ with 
    $j=1,2,\ldots, m$ and $m \leq E(T)+1$ such that 
    \begin{align}
        \label{eq:dimension_sum_trimmed}
        \sum_{j=1}^{m} (k_j - 2) =l\,,
    \end{align}
    and
    \begin{equation}
        \label{eq:distance_orth_trimmed}
	       d^2(T_1, T_2) = \sum_{j=1}^{m} d^2(T_1^j, T_2^j) \,. 
    \end{equation}
    Another application of Lemma \ref{lemma:shared_edges} for $T$ and $T_i$, $i=1,2$, since $\cE = E(T)$ and hence $T^j = \stree \in \TT_{k_j}$ for all $j=1,\ldots,m$ yields at once
    \begin{align}
        \label{eq:distance_orth_reference}
	       d^2(T, T_i) = \sum_{j=1}^{m} d^2(\stree, T_i^j) \,. 
    \end{align}
    We have by \eqref{eq:bhv-distance-explicit} that 
    \begin{align}
    \label{eq:scale_dist_to_T}
		d(T, \gamma_i^j(\lambda)) = \sqrt{\sum_{s \in E(T_i)\setminus E(T)} \lvert s \rvert^2_{\gamma_i(\lambda))}} = \lambda \cdot \sqrt{\sum_{s \in E(T_i) \setminus E(T)} \lvert s \rvert^2_{T_i}} = \lambda \cdot d(T, T_i)\,, \quad i=1,2\,.
	\end{align}
    Next, let $(\mathcal{A}, \mathcal{B})$ with $\mathcal{A} = (A_0, \ldots, A_k)$, $\mathcal{B} =(B_0, \ldots, B_k)$ 
    be the support pair of the geodesic between $T_1,T_2$. 
    Note that since $E(T) \subset E(T_i) \subseteq C(T)$ for $i=1,2$, for every $\lambda \in (0,1]$, one has $E(\gamma_i(\lambda)) = E(T_i)$, $i=1,2$. Then, we have for every $\lambda \in (0,1]$ and $j = 1,2,\ldots,k$
    \begin{align}\label{eq:pull-out-lambda}
        \lVert A_j \rVert_{\gamma_1(\lambda)} = \lambda \cdot \lVert A_j \rVert_{T_1} \text{ and } \lVert B_j \rVert_{\gamma_2(\lambda)} = \lambda \cdot \lVert B_j \rVert_{T_2}\,,
    \end{align}
    and consequently
    \[
        \frac{\lVert A_r \rVert_{\gamma_1(\lambda)}}{\lVert B_r \rVert_{\gamma_2(\lambda)}} = \frac{\lVert A_r \rVert_{T_1}}{\lVert B_r \rVert_{T_2}}\,,\quad r=1,\ldots,k\,.
    \]
    Now, \eqref{eq:bhv-geodesics1} and \eqref{eq:bhv-geodesics2}, similarly \eqref{eq:bhv-geodesics2} for $T_1, T_2$ implies at once its validity for $\gamma_1(\lambda),\gamma_2(\lambda)$ for all $\lambda \in  (0,1]$, yield that $(\mathcal{A}, \mathcal{B})$ must also be the support pair of the geodesic between $\gamma_1(\lambda)$ and $\gamma_2(\lambda)$ for all $\lambda \in (0,1]$. Again by \eqref{eq:bhv-distance-explicit}, we have 
    \begin{eqnarray}
    \label{eq:scale_dist_between}
      \nonumber
		d(\gamma_1(\lambda), \gamma_2(\lambda)) &=& \sqrt{\sum_{s \in A_0}  (\lvert s \rvert_{\gamma_2(\lambda)} - \lvert s \rvert_{\gamma_1(\lambda)})^2 +  \sum_{r =1}^k(\lVert A_r \rVert_{\gamma_1(\lambda)} + \lVert B_r \rVert_{\gamma_2(\lambda)})^2}\\
  \nonumber
  &\stackrel{\eqref{eq:same_lengths}}{=}&\sqrt{\sum_{s \in A_0 \setminus E(T)}  (\lvert s \rvert_{\gamma_2(\lambda)} - \lvert s \rvert_{\gamma_1(\lambda)})^2 +  \sum_{r =1}^k(\lVert A_r \rVert_{\gamma_1(\lambda)} + \lVert B_r \rVert_{\gamma_2(\lambda)})^2}\\
  \nonumber
  &\stackrel{\eqref{eq:pull-out-lambda}}{=}& \lambda \cdot \sqrt{\sum_{s \in A_0 \setminus E(T)} (\lvert s \rvert_{T_1} - \lvert s \rvert_{T_1})^2 +  \sum_{r =1}^k(\lVert A_r \rVert_{T_1} + \lVert B_r \rVert_{T_2})^2}\\
  &=& \lambda \cdot d(T_1,T_2)\,.
            \end{eqnarray}

    Let $\sigma_i^j = \dir_\stree(T_i^j)$ for $i=1,2$ and $j=1,2,\ldots,m$. Then 
    \begin{eqnarray}
		\label{eq:angle_decomp}
        \nonumber
		\cos(\angle_T(\sigma_1, \sigma_2)) & =& \lim_{\lambda\searrow 0}
		\frac{d^2(T, \gamma_1(\lambda)) + d^2(T, \gamma_2(\lambda)) - d^2(\gamma_1(\lambda), \gamma_2(\lambda))}{2 d(T, \gamma_1(\lambda)) d(T, \gamma_2(\lambda))}                                                                                                        \\
      \nonumber
                                          & \stackrel{\eqref{eq:scale_dist_to_T} \text{ and } \eqref{eq:scale_dist_between}}{=}& \frac{d^2(T, T_1) + d^2(T, T_2) - d^2(T_1, T_2)}{2 d(T, T_1) d(T, T_2)}\\ 
      \nonumber
		                                   & \stackrel{\eqref{eq:distance_orth_trimmed} \text{ and } \eqref{eq:distance_orth_reference}}{=}&\sum_{j=1}^{m} \frac{d(\stree, T_1^j)}{d(T, T_1)} \frac{d(\stree, T_2^j)}{d(T, T_2)}
		\frac{d^2(\stree, T_1^j) + d^2(\stree, T_2^j) - d^2(T_1^j, T_2^j)}{2 d(\stree, T_1^j) d(\stree, T_2^j)}                                                                             \\
		                                   & =& \sum_{j=1}^{m} \frac{d(\stree, T_1^j)}{d(T, T_1)} \frac{d(\stree, T_2^j)}{d(T, T_2)}\cdot \cos\left(\angle_\stree\left(\sigma_1^j, \sigma_2^j\right)\right)
    \end{eqnarray}

    Now, set for $i \in \{1,2\}\,, j \in \{1,2,\ldots,m-1\}$
    \begin{align}
    \label{eq:eta}
    	\eta_i^j = \begin{cases}
    	    0 &\text{if } d(\stree, T_i^j) = 0\,,\\
               \arccos\left(\frac{d(\stree, T_i^j)}{\sqrt{d^2(T,T_i) - \sum_{r<j}d^2(\stree, T_i^r)}}\right) \quad &\text{ else.}
    	\end{cases}
    \end{align}
    Note that $\eta_i^j \in [0,\pi]$ since
    \begin{align*}
    d^2(\stree, T_i^j) + \sum_{r<j}d^2(\stree, T_i^r) = \sum_{r=1}^j d^2(\stree, T_i^r)
    \leq
    \sum_{r=1}^m d^2(\stree, T_i^r) \stackrel{\eqref{eq:distance_orth_reference}}{=} d^2(T,T_i) \\\Leftrightarrow \\
    \frac{d(\stree, T_i^j)}{\sqrt{d^2(T,T_i) - \sum_{r<j}d^2(\stree, T_i^r)}} \leq 1\,.
    \end{align*}
    
    Hence, we have 
    \[
    \sin(\eta_i^j) = \sqrt{1 - \cos^2(\eta_i^j)} = \sqrt{1 - \frac{d^2(\stree, T_i^j)}{d^2(T,T_i) - \sum_{k<j}d^2(\stree, T_i^k)}} = \sqrt{\frac{d^2(T,T_i) - \sum_{k\leq j}d^2(\stree, T_i^k)}{d^2(T,T_i) - \sum_{k< j}d^2(\stree, T_i^k)}}
    \]
    Then, we have for $i=1,2$ and $1\leq j < m$:
    \begin{align*}
        \cos(\eta_{i}^j) \cdot \prod_{k=1}^{j-1} \sin(\eta_i^k) &= \frac{d(\stree, T_i^j)}{\sqrt{d^2(T,T_i) - \sum_{\ell<j}d^2(\stree, T_i^\ell)}} \cdot \prod_{k=1}^{j-1}\sqrt{\frac{d^2(T,T_i) - \sum_{\ell\leq k}d^2(\stree, T_i^\ell)}{d^2(T,T_i) - \sum_{\ell< k}d^2(\stree, T_i^\ell)}}\\
        &= \frac{d(\stree, T^j_i)}{d(T,T_i)}\,,
    \end{align*}  
    and 
    \begin{align*}
        \prod_{k=1}^{m-1} \sin(\eta_i^k) = \frac{\sqrt{d^2(T,T_i) - \sum_{\ell\leq m-1}d^2(\stree, T_i^\ell)}}{d(T,T_i)} \stackrel{\eqref{eq:distance_orth_reference}}{=} \frac{d(\stree, T^m_i)}{d(T,T_i)} \,,
    \end{align*}
    Thus, \eqref{eq:angle_decomp} becomes
	\begin{align}\label{eq:join-distance-proof-theorem}\nonumber
		\cos(\angle_T(\sigma_1, \sigma_2)) & =  \sum_{j=1}^{m} \left(\prod_{k=1}^{j-1}\sin(\eta_1^k) \sin(\eta_2^k)\right)\cos(\eta_1^j)\cos(\eta_2^j) \cdot \cos\left(\angle_\stree\left(\sigma_1^j, \sigma_2^j\right)\right) \\
		                                   & \quad+ \left(\prod_{k=1}^{m-1} \sin(\eta_1^k) \sin(\eta_2^k)\right) \cdot \cos\left(\angle_\stree\left(\sigma_1^m, \sigma_2^m\right)\right) \,,  
	\end{align}
     revealing the structure from (\ref{eq:join-distance}) of a nested  spherical join:
    \[
        \perpdir \cong \mathbb{L}_{k_1} \ast \Bigg( \mathbb{L}_{k_2} \ast \Big(\mathbb{L}_{k_3} \ast \big(\cdots  \ast( \mathbb{L}_{k_{m-1}} \ast \mathbb{L}_{k_m})\cdots \big)\Big) \Bigg) \,.
    \]
	The structure of the tangent cone follows from \cite[Proposition I.5.15]{bridson}. Let $C_0 M$ denote the Euclidean cone over a metric space $M$, c.f. \cite[Definiton I.5.6]{bridson}.
    Then, by 
    \begin{align*}
        \mathfrak{T}_T = C_0 \Sigma_T &\cong C_0 S^{N-l-3} \times C_0 \perpdir \\
        &\cong \RR^{N-l-2} \times  C_0\left(\mathbb{L}_{k_1} \ast \Bigg( \mathbb{L}_{k_2} \ast \Big(\mathbb{L}_{k_3} \ast \big(\cdots  \ast( \mathbb{L}_{k_{m-1}} \ast \mathbb{L}_{k_m})\cdots \big)\Big) \Bigg)\right) \\
         &\cong \RR^{N-l-2} \times C_0\mathbb{L}_{k_1} \times C_0\Bigg( \mathbb{L}_{k_2} \ast \Big(\mathbb{L}_{k_3} \ast \big(\cdots  \ast( \mathbb{L}_{k_{m-1}} \ast \mathbb{L}_{k_m})\cdots \big)\Big) \Bigg) \\
         &\ \vdots\\
         & \cong \RR^{N-l-2} \times C_0\mathbb{L}_{k_1} \times \cdots \times C_0\mathbb{L}_{k_1}\\
         &\stackrel{\text{Proposition \ref{prop:bhv_cone}}}{\cong} \RR^{N-l-2} \times \TT_{k_1} \times \cdots \times \TT_{k_m}\,.
    \end{align*}

Now, let us pick $\sigma_1 := \sigma =  (\eta_1^1, \ldots, \eta_1^{m-1}, \sigma_1,\ldots,\sigma_m) \in \perpdir$.
As an element in $\Sigma_T$, it is given by $(0, \sigma_1^\parallel, \sigma_1)$, where $\sigma_1^\parallel \in \Sigma_T^\parallel$
is arbitrary. 

Next, we decompose down $\sigma_2 := \dir_T(T, T^\prime)$. 
Let $\mathfrak{T}^\perp_T := C_0 \perpdir $ and let $\varpi_\perp: 
\mathfrak{T}_T \to \mathfrak{T}_T^\perp$ denote the canonical projection.
Writing $\sigma_2$ as element of $\Sigma^\parallel\ast \Sigma_\perp$, one has
$\sigma_2 = (\eta_0^i, \sigma_2^\parallel, \sigma_2^\perp)$, where for the origin $\cO \in \mathfrak{T}_T^\perp$
\[
    \eta_0^1 = \arccos\left(\frac{d(\cO,\varpi_\perp(\log_T(T^\prime)))}{d(T, T^\prime)}\right)\,.
\]
Then, we obtain by the definition of spherical joins 
\begin{eqnarray}
\label{eq:perp_pull}
\nonumber
   d(T, T^\prime) \cdot \cos(\angle_T(\sigma, T^\prime)) &=& d(T, T^\prime) \cos(\eta_0^1) \cdot \cos(\angle_\cO(\sigma, \varpi_\perp(T^\prime))) \\
   &=& d(\cO,\varpi_\perp(\log_T(T^\prime))) \cdot \cos(\angle_\cO(\sigma, \varpi_\perp(T^\prime)))
\end{eqnarray}

At last, we prove (ii). As $\mathfrak{T}^\perp_T \cong \TT_{k_1} \times \cdot \TT_{k_m}$, we have
\[  
    d(\cO, \varpi_\perp(\log_T(T^\prime))  = \sum_{i=1}^m d^2(\stree, \varpi_i(\log_T(T^\prime)))\,,
\]
The orthogonal directions $\perpdir$ correspond to the addition of splits 
that are compatible with the topology. In particular, $\sigma_2^\perp$ 
corresponds to the the addition of splits in $T^\prime$ that are compatible 
with the splits of $T$.

Thus, we can identify  $\varpi_\perp(\log_T(T^\prime))$ with the tree 
${T'}^\perp := T + \sum_{s \in C(T) \setminus E(T)} \lvert s \rvert_{T'} \cdot s$. 

In particular, we have $\sigma_2^\perp = \dir_T$ and $d(\cO, \varpi_\perp(\log_T(T^\prime)) 
= d(T, {T'}^\perp)$. 
Another application of Lemma \ref{lemma:shared_edges} for $T, {T^\prime}^\perp$
and $\cE = E(T)$ then yields the explicit trees ${T^\prime}^i = \varpi_j(\log_\mean(T^\prime)) \in \TT_{k_j}$, 
$j =1,\ldots,m$, and we have, as before,
\[\sigma_2^\perp = (\eta_2^1, \ldots, \eta_2^{m-1}, \dir_\stree({T^\prime}^1, \ldots, {T^\prime}^m)\,,\]
where for $j=1,\ldots,m$
\begin{align*}
    	\eta_2^j = \begin{cases}
    	    0 &\text{if } d(\stree, (T^j) = 0\,,\\
               \arccos\left(\frac{d(\stree, {T^\prime}^j)}{\sqrt{d^2(T,{T^\prime}^\perp) - \sum_{r<j}d^2(\stree, {T^\prime}^j)}}\right) \quad &\text{ else.}
    	\end{cases}
\end{align*}
Finally, we obtain 
\begin{eqnarray*}
   \frac{d(T, T^\prime)}{d(T,\varpi_\perp(\log_T(T^\prime)))} &\cdot& \cos(\angle_T(\sigma,wT^\prime)) \stackrel{\eqref{eq:perp_pull}}{=} 
   \cos(\angle_\cO(\sigma, \varpi_\perp(T^\prime)))\\
   &\stackrel{\eqref{eq:join-distance-proof-theorem}}{=}& \sum_{j=1}^{m} \left(\prod_{k=1}^{j-1}\sin(\eta_1^k) \sin(\eta_2^k)\right)\cos(\eta_1^j)\cos(\eta_2^j) \cdot \cos\left(\angle_\stree\left(\sigma_1^j, \sigma_2^j\right)\right) \\
		                                   &&+ \left(\prod_{k=1}^{m-1} \sin(\eta_1^k) \sin(\eta_2^k)\right) \cdot \cos\left(\angle_\stree\left(\sigma_1^m, \sigma_2^m\right)\right) \,,  \\ 
		 &=& \sum_{i=1}^{m-1} \left(\prod_{j=1}^{i-1} \sin(\eta_j)\right)\cos(\eta_i) \cdot \frac{d(\stree, \varpi_i(\log_T(T^\prime)))}{d(T,\varpi_\perp(\log_T(T^\prime)))} \cdot \cos(\angle_{\stree}(\sigma_i, \varpi_i(\log_T(T^\prime)))) \\
	&& + \left(\prod_{j=1}^{m-1} \sin(\eta_j)\right) \cdot \frac{d(\stree, \varpi_m(\log_T(T^\prime)))}{d(T,\varpi_\perp(\log_T(T^\prime)))} \cdot \cos(\angle_{\stree}(\sigma_m, \varpi_m(\log_T(T^\prime)))) \,.
\end{eqnarray*}
Rearranging yields the assertion of (ii).

%

\subsection{Proof of Lemma \ref{lem:prox}}

 (i): 
 If $\angle_\stree(\tau, T) =0$, i.e. $\tau = \dir_\stree T$. Since 
 $$ f_T(\sigma) \geq -1 = f_T(\dir_\stree (T))\,,\quad  \angle_\stree(\sigma, \tau) \geq 0 = \angle_\stree(\sigma, \tau)$$
 for all $\sigma \in \Sigma_\stree$, we have at once that $\tau = \dir_\stree (T)$ is the unique element in $\prox_T(\tau)$.

 Now assume that $\angle_\stree(\tau, T) >0$, i.e. $\tau \neq \dir_\stree T$. Assume that $\sigma \in \prox_T(\tau)$ with $\lambda = \angle_\stree(\sigma, \tau)$.
 
 Case I: $\lambda > \angle_\stree(T, \tau)$. Then
 $$f_T(\sigma) + \frac{\angle_\stree(\sigma, \tau)^2}{2\nu} \geq - \cos (0) + \frac{\lambda^2}{2\nu}> f_T(\dir_\stree(T)) + \frac{\angle_\stree(T, \tau)^2}{2\nu}$$
a contradiction to $\sigma \in P_T(\tau)$.

 Hence we are in 
 Case II: $\lambda \leq \angle_\stree(T, \tau) \leq \pi$.
 Since $\bar\beta^{\dir_\stree (T)}_\tau$ is a geodesic, we have by the triangle inequality that
\begin{eqnarray*}
\lambda + \angle_\stree(\bar\beta^{\dir_\stree (T)}_\tau(\lambda) ,T) &=& \angle_\stree(\tau, \bar\beta^{\dir_\stree (T)}_\tau(\lambda))+\angle_\stree(\bar\beta^{\dir_\stree (T)}_\tau(\lambda) ,T)\\
&=& \angle_\stree(\tau,T)\\
&\leq& \angle_\stree(\tau, \sigma) + 
\angle_\stree(\sigma,T)\\
&=&\lambda + \angle_\stree(\sigma,T)\,.
\end{eqnarray*}
This yields
\begin{eqnarray*}
 f_T(\sigma) + \frac{\angle_\stree(\sigma, \tau)^2}{2\nu}&=& -\cos \angle_\stree(T,\sigma) + \frac{\lambda^2}{2\nu}\\&\geq&  -\cos \angle_\stree(\bar\beta^{\dir_\stree (T)}_\tau(\lambda),T) + \frac{\angle_\stree(\bar\beta^{\dir_\stree (T)}_\tau(\lambda), \tau)^2}{2\nu}
\end{eqnarray*}
and hence equality above due to $\sigma \in \prox_T(\tau)$. Thus $\bar\beta^{\dir_\stree (T)}_\tau(\lambda) \in \prox_T(\tau)$ and $\lambda$ can be obtained by minimizing 
\begin{eqnarray}\label{eq:geodesic-prox-op}
 f_T(\bar\beta^{\dir_\stree (T)}_\tau(\lambda)) + \frac{\lambda^2}{2\nu}\,.
 \end{eqnarray}
In particular, if $\angle_\stree(T, \tau) <\pi $, since  $\bar\beta^{\dir_\stree (T)}_\tau(\lambda)$ has the same distance to $\tau$ and $T$ respectively, as $\sigma$, due to uniqueness of geodesics of length $<\pi$ (see Remark \ref{rm:Cat-1}), we have uniqueness
$$\left\{\bar\beta^{\dir_\stree (T)}_\tau(\lambda)\right\} = \prox_T(\tau)\,.$$

Moreover, if $\angle_\stree(T, \tau) =\pi $, then 
$\angle_\stree(\bar\beta^{\dir_\stree (T)}_\tau(\lambda), T)= \pi - \lambda$, hence $ f_T(\bar\beta^{\dir_\stree (T)}_\tau(\lambda)) = - \cos (\pi-\lambda)= \cos \lambda$ and thus
\begin{eqnarray*}
\frac{d^2}{d\lambda^2} \left(f_T(\bar\beta^{\dir_\stree (T)}_\tau(\lambda)) + \frac{\lambda^2}{2\nu}\right)  &=& -\cos \lambda + \frac{1}{\nu}\,.
\end{eqnarray*}
By hypothesis, $\nu \leq 1$, hence the above is positive for all $0<\lambda<2\pi $, and thus (\ref{eq:geodesic-prox-op}), for $0\leq \lambda \leq \pi$, is uniquely minimized at $\lambda =0$, i.e. $\{\tau\}=  P_T(\tau)$.

(ii): W.l.o.g. assume $\|T\| = 1$.
	Since $\stree$ is the cone point of $\bhv$, we have $\angle_{\stree}(T_\tau, T) =
		\arccos\left(\frac{d^2(\stree, T) + d^2(\stree, T) - d^2(T_\tau, T)}{2d(\stree, T_\tau) d(\stree, T)}\right)$.
	Consequently, the convex hull of the geodesic triangle spanned by $\stree, T_\tau, T$ is
	isometric to the convex hull of a triangle in $\RR^2$ with equal edge lengths, see
	\cite[Proposition II.2,9]{bridson}, and we can utilize Euclidean geometry
	for the proof.
 
	Abbreviating $\theta := \angle_\stree(\sigma, \tau)$, $\beta := \overline{\beta}_\sigma^\tau$ and $\gamma := \overline{\gamma}_{T_\tau}^{T}$ 
    we thus search for $\lambda^\prime \in [0,1]$, given  $\lambda \in [0,1]$, such that 	$\lambda \cdot \theta = \angle_\stree\left(\sigma, \gamma(\lambda^\prime)\right)$ (illustrated inFigure \ref{fig:sphdist}).
    
    \begin{figure}
		\centering
		\includegraphics[scale=1]{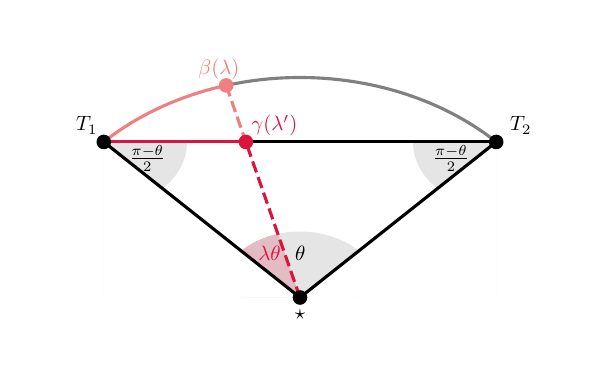}
		\caption{The geodesic triangle in the proof of \ref{lem:prox}. Here, we identify $\Sigma_\stree$ with the link $\mathbb{L}_n$ by Proposition~\ref{prop:bhv_cone}, and then construct the geodesic $\beta$ by reparametrizing $\gamma$ 
        and projecting back to the link.
        }
		\label{fig:sphdist}
	\end{figure}
	
 Using the law of sines, we obtain
	\begin{align*}
		\lambda^\prime \cdot d(T_\tau, T) = d\left(T_\tau, \gamma(\lambda^\prime)\right) = \frac{\sin(\lambda \cdot \theta)}{\sin\left(\pi - \left(\lambda\theta + \frac{\pi - \theta}{2}\right) \right)}\,.
	\end{align*}
	Since, on the other hand,
	\begin{align*}
		d(T_\tau,T) = \sqrt{d(\stree, T_\tau) + d(\stree, T) - 2 \cos(\theta)} =  \sqrt{ \left(2  - e^{i \theta} - e^{- i \theta}\right)} = \sqrt{ \left(e^{i \frac{\theta}{2}} - e^{- i \frac{\theta}{2}}\right)^2} = 2 \sin\left(\frac{\theta}{2}\right)\,,
	\end{align*}
	we have the assertion
	\begin{align*}
		\lambda^\prime = \frac{\sin(\lambda \cdot \theta)}{2 \sin\left(\frac{\theta}{2}\right) \sin\left(\frac{\pi+(1-\lambda)\theta}{2}\right)} 
    \,.
	\end{align*}
\qed

\end{document}